\documentclass[journal]{IEEEtran}
\ifCLASSINFOpdf
\else
\fi
\usepackage{amsmath}
\usepackage{amssymb}  
\usepackage{mathtools}

\usepackage{mleftright} 

\usepackage{amsthm} 

\usepackage{color,soul}

\usepackage{tikz}
\usepackage{pgfplots}
\usetikzlibrary{positioning}
\usetikzlibrary{calc}
\usetikzlibrary{decorations.markings}
\usetikzlibrary{arrows}
\usepgfplotslibrary{colorbrewer}
\usepgfplotslibrary{patchplots}
\pgfplotsset{compat=newest}
\usepackage{pbox} 
\usepackage{relsize}

\usepackage{setspace}

\definecolor{A}{RGB}{223,189,217}
\definecolor{B}{RGB}{223,179,217}
\definecolor{C}{RGB}{213,159,217}
\definecolor{D}{RGB}{174,134,202}
\definecolor{E}{RGB}{174,124,202}
\definecolor{F}{RGB}{174,114,202}
\definecolor{G}{RGB}{168,114,203}
\definecolor{H}{RGB}{148,114,203}
\definecolor{I}{RGB}{122,109,198}
\definecolor{J}{RGB}{102,99,188}
\definecolor{K}{RGB}{81,88,179}
\definecolor{L}{RGB}{71,70,160}
\definecolor{M}{RGB}{51,60,154}
\definecolor{N}{RGB}{35,50,124}
\definecolor{O}{RGB}{26,50,104}
\definecolor{P}{RGB}{1,1,1}

\definecolor{lilla}{HTML}{750787}
\definecolor{mycolor1}{rgb}{0.00000,0.44700,0.74100}%

\newtheorem{theorem}{\bf Theorem} 
\newtheorem{definition}{\bf Definition} 
\newtheorem{lemma}{\bf Lemma} 
\newtheorem{remark}{\bf Remark}

\newtheorem{assumption}{\bf Assumption} 
\newtheorem{example}{\bf Example}

\begin{document}
\title{
Provably robust verification of dissipativity properties from data
}

\author{Anne~Koch, %
				Julian~Berberich, %
        and~Frank~Allg\"ower%
\thanks{A Koch, J Berberich and F Allg\"ower are with the Institute for Systems Theory and Automatic Control, University of Stuttgart. 
This work was funded by Deutsche Forschungsgemeinschaft (DFG, German
Research Foundation) under Germany’s Excellence Strategy - EXC 2075 -
390740016. The authors thank the International Max Planck Research School
for Intelligent Systems (IMPRS-IS) for supporting Anne Koch and Julian Berberich. 
E-mail: {\tt\small$\{$anne.koch, julian.berberich, frank.allgower$\}$@ist.uni-stuttgart.de}}}%

\maketitle

\begin{abstract}
Dissipativity properties have proven to be very valuable for systems analysis and controller design. With the rising amount of available data, there has therefore been an increasing interest in determining dissipativity properties from (measured) trajectories directly, while an explicit model of the system remains undisclosed. Most existing approaches for data-driven dissipativity, however, guarantee the dissipativity condition only over a finite time horizon and provide weak or no guarantees on robustness in the presence of noise. In this paper, we present a framework for verifying dissipativity properties from measured data with desirable guarantees. We first consider the case of input-state measurements, where we provide non-conservative and computationally attractive conditions in the presence of noise. We extend this approach to input-output data, where similar results hold in the noise-free case.
We then provide results for the noisy input-output data case, which is particularly challenging. Finally, we apply the proposed approach in a real-world experiment and illustrate its applicability and advantages compared to established methods based on system identification.
\end{abstract}


\IEEEpeerreviewmaketitle

\section{Introduction}
\IEEEPARstart{W}{ith} the rising complexity of systems, obtaining a suitable mathematical model for a yet unknown systems becomes more and more cumbersome. At the same time, data is becoming ubiquitous and cheap. Therefore, there has been a rising interest in establishing a data-driven framework that allows for systems analysis and control from data with the same guarantees as obtained through the well-known and established model-based approaches. Especially for linear time-invariant (LTI) systems, there has recently been considerable progress in setting up such a data-driven framework. The basis for this line of work can be attributed to the seminal work in \cite{Willems05}, in which the authors prove in the behavioral framework that the behavior of an LTI system can be described by suitable data-dependent matrices under the condition that the input is persistently exciting. This representation in state-space as stated and discussed in \cite{Berberich2019a} and proven 
in \cite{Waarde2020}, provides a basis that allows for systems analysis and controller design with rigorous guarantees from (measured) trajectories. Recent developments in this direction include state-feedback design from input-state trajectories \cite{Persis2019}, robust controller synthesis from noisy input-state trajectories \cite{Berberich2019c}, data-driven model predictive control \cite{Coulson2019,Berberich2019b}, data informativity~\cite{Waarde2020a}, dissipativity properties from input-output trajectories \cite{Maupong2017,Romer2019a,Koch2020} and from input-state trajectories \cite{Koch2020a}.

As in \cite{Maupong2017,Romer2019a,Koch2020,Koch2020a}, we are interested in dissipativity properties from data. Dissipativity properties cannot only be used for systems analysis giving insights into an unknown system, but knowledge of dissipativity properties allows for direct application of well-known feedback theorems with guaranteed stability of the closed loop. For examples of such feedback theorems and stabilizing, robust or distributed controller design on the basis of dissipativity properties, the reader is referred to the standard literature with respect to dissipativity properties, which includes \cite{Zames1966, Desoer1975, Schaft2000}. Due to the well-established literature on dissipativity-based controller design, there has been a considerable number of approaches to determine such dissipativity properties from data.

Very generally, the literature on data-driven dissipativity can be roughly categorized into three inherently different setups. Firstly, a large number of approaches consider online sampling schemes for LTI systems, where it is assumed that it is possible to choose the input and measure the output in an iterative fashion. This line of works includes \cite{Wahlberg2010,Rojas2012,Tanemura2019,Romer2019c,Mueller2017}. While there are some distinct advantages and disadvantages to each of these methods, the joint limitation is that iterative experiments are needed, which requires access to the plant and is potentially a more time-consuming task than purely computational and offline approaches. 
We define the second category as approaches that apply for rather general classes of nonlinear systems, but require large or even huge amounts of input-output trajectories (e.g.~\cite{Montenbruck2016a,Romer2017a,Sharf2020,Romer2019b,Martin2020}). While these works consider more general nonlinear systems, the sheer amount of required input-output trajectories hampers their application. 

Finally, the third category includes all offline computational approaches from one input-state or input-output trajectory for LTI systems, which includes \cite{Maupong2017,Romer2019a,Koch2020,Koch2020a,Saeki2020}. The approaches in \cite{Maupong2017,Romer2019a,Koch2020,Saeki2020} do not provide guarantees from noisy trajectories. 
Therefore, we extend in this work the idea presented in \cite{Koch2020a}, where rigorous and quantitative guarantees from noise-corrupted input-state trajectories can be given.
However, the therein presented result are generally not tight and also the computational complexity grows with increasing amounts of data. 
In this paper, we employ ideas similar to recent results on data-driven controller design in~\cite{Waarde2020b,Berberich2020} in order to derive both non-conservative and computationally attractive conditions for data-driven dissipativity.

The remainder of the paper is structured as follows. In Sec.~\ref{sec:setup}, we introduce the problem formulation and present some related results that will be used throughout the paper. We then introduce an equivalent dissipativity characterization purely on the basis of input-state data in Sec.~\ref{sec:diss} followed by a noisy consideration thereof in Sec.~\ref{sec:noise}, where we provide a tight robust verification framework for dissipativity properties. Next, we extend the results to input-output trajectories first in the noise-free case in Sec.~\ref{sec:output} followed by a consideration of noise-corrupted trajectories in Sec.~\ref{sec:output_noise}. Finally, we apply the introduced approaches to real-world data of a two-tank water system.

\section{Problem setup}
\label{sec:setup}
We consider multiple-input multiple-output 
discrete-time LTI systems for which there exists a (controllable) minimal realization of the form
\begin{align}
\begin{split}
x_{k+1}&=Ax_k+Bu_k, \>\>x_0=\bar{x},\\
y_k&=Cx_k+Du_k,
\label{eq:sys}
\end{split}
\end{align}
with $x_k \in \mathbb{R}^{n}$, $u_k \in \mathbb{R}^m$ and $y_k \in \mathbb{R}^p$.  

In this paper, we develop a framework for verifying dissipativity properties of~\eqref{eq:sys} directly from measured data, without identifying a model of the system.
We thereby consider two cases:
\begin{itemize}
\item \textbf{Input-state data} (Sec.~\ref{sec:diss} \& \ref{sec:noise})\textbf{:}
We assume that $A$ and $B$ are unknown, but one input-state trajectory $\{x_k\}_{k=0}^N$, $\{u_k\}_{k=0}^{N-1}$ is available. 
Further, we assume that\footnote{It is straightforward to extend the presented results to the case that $C$ and $D$ are unknown but measurements of $\{y_k\}_{k=0}^{N-1}$ are available.} $C$, $D$ are known.
\item \textbf{Input-output data} (Sec.~\ref{sec:output} \& \ref{sec:output_noise})\textbf{:} We assume that $A,B,C$ and $D$ are unknown, but one input-output trajectory $\{u_k\}_{k=0}^{N-1}$, $\{y_k\}_{k=0}^{N-1}$ is available as well as an upper bound on the lag of the system $l \geq \underline{l}$ (cf.\ Def.~\ref{def:lag}).
\end{itemize}
For each case, we in turn distinguish between noise-free measurements of~\eqref{eq:sys} (Sec.~\ref{sec:diss} \& \ref{sec:output}) and noisy data (Sec.~\ref{sec:noise} \& \ref{sec:output_noise}).
We collect the respective data sequences
$\{u_k\}_{k=0}^{N-1}$, $\{x_k\}_{k=0}^{N}$ or $\{y_k\}_{k=0}^{N-1}$ in the following matrices 
\begin{align*}
X &\coloneqq \begin{pmatrix} x_0 & x_1 & \cdots & x_{N-1} \end{pmatrix}, \\
X_+ &\coloneqq \begin{pmatrix} x_1 & x_2 & \cdots & x_N \end{pmatrix}, \\
U &\coloneqq \begin{pmatrix} u_0 & u_1 & \cdots & u_{N-1} \end{pmatrix}, \\
Y &\coloneqq \begin{pmatrix} y_0 & y_1 & \cdots & y_{N-1} \end{pmatrix}.
\end{align*}

Our approach is based on data, i.e., on measured trajectories of the system~\eqref{eq:sys}, with the only assumption that this measured trajectory is informative enough. 
One condition in this respect, which will play an important role in the following sections, is the rank condition
\begin{align}
\mathrm{rank} \begin{pmatrix} X \\ U \end{pmatrix} = n+m.
\label{eq:rank}
\end{align}
Generally speaking, this condition can be ensured by requiring that the input of the measured trajectory is sufficiently persistently exciting \cite{Willems05}. Given a finite sequence $\left\{u_k\right\}_{k=0}^{N-1}$, we first define the corresponding Hankel matrix
\begin{align*}
H_L&(u)\coloneqq\begin{pmatrix}u_0 & u_1 & \dots & u_{N-L}\\
u_1 & u_2 & \dots & u_{N-L+1}\\
\vdots & \vdots & \ddots & \vdots\\
u_{L-1} & u_L & \dots & u_{N-1}
\end{pmatrix}.
\end{align*}
We can now recall the notion of persistency of excitation.
\begin{definition}\label{def:pe}
We say that a sequence $\left\{u_k\right\}_{k=0}^{N-1}$ with $u_k\in\mathbb{R}^m$ is persistently exciting of order $L$, if $\text{rank}\left(H_L(u)\right)=mL$.
\end{definition}
With this definition of persistency of excitation, we can find a sufficient condition to ensure \eqref{eq:rank}.
\begin{lemma}[\cite{Willems05},~Corollary~2]
\label{lem:pe}
If the sequence $\left\{u_k\right\}_{k=0}^{N-1}$ with $u_k\in\mathbb{R}^m$ is persistently exciting of order $n+1$, 
then condition \eqref{eq:rank} holds.
\end{lemma}

Since their introduction in \cite{Willems1972}, dissipativity properties have become increasingly relevant in systems analysis and control. Usually, these properties can be verified using a full mathematical model of the system.
In this paper, we are interested in determining dissipativity properties directly from (noisy) data with guarantees. While the notion of dissipativity was introduced in \cite{Willems1972} for general (nonlinear) systems, we make use of equivalent formulations for LTI systems with quadratic supply rates as, e.g., presented in \cite{Scherer2000}. Quadratic supply rates are functions $s: \mathbb{R}^m \times \mathbb{R}^p \rightarrow \mathbb{R}$ defined by 
\begin{align}
s(u,y) = \begin{pmatrix} u \\ y \end{pmatrix}^\top \Pi
\begin{pmatrix} u \\ y \end{pmatrix}.
\label{eq:supply} 
\end{align}
The matrix $\Pi \in \mathbb{R}^{(m+p) \times (m+p)}$ will be partitioned throughout this paper as
\begin{align*}
\Pi=\begin{pmatrix}R &S^\top\\S&Q\end{pmatrix}
\end{align*}
 with $Q=Q^{\top} \in \mathbb{R}^{p\times p}$, $S \in \mathbb{R}^{p\times m}$ and $R=R^\top \in \mathbb{R}^{m\times m}$. 
\begin{definition}
\label{def:1}
A system \eqref{eq:sys} is said to be dissipative w.r.t. the supply rate $s$
if there exists a function $V: \mathbb{R}^n \rightarrow \mathbb{R}$ which is bounded from below such that
\begin{align}
V(x_{k^{\prime \prime}}) - V(x_{k^\prime}) \leq \sum_{k=k^\prime}^{k^{\prime \prime} -1} s(u_k,y_k)
\label{eq:diss-nonstrict}
\end{align}
for all $0 \leq k^\prime < k^{\prime \prime}$ and all signals $(u,x,y)$ which satisfy \eqref{eq:sys}.
It is said to be \textit{strictly} dissipative if instead of \eqref{eq:diss-nonstrict}
\begin{align*}
V(x_{k^{\prime \prime}}) - V(x_{k^\prime}) \leq \sum_{k=k^\prime}^{k^{\prime \prime} -1} s(u_k,y_k) - \epsilon \sum_{k=k^\prime}^{k^{\prime \prime} -1} \|u_k\|_2^2
\end{align*}
holds for all $0 \leq k^\prime < k^{\prime \prime}$, all signals $(u,x,y)$ which satisfy \eqref{eq:sys} and some $\epsilon > 0$.
\end{definition}

Hereby, the matrices $(Q,S,R)$ in the supply rate define the system property at hand. With the supply rates defined by
\begin{align}
\Pi_{\gamma} = \begin{pmatrix} \gamma^2 I & 0 \\ 0 & -I \end{pmatrix}, \quad \Pi_{\text{P}} = \begin{pmatrix} -\rho I & 0.5 I \\ 0.5 I & 0 \end{pmatrix},
\label{eq:pi}
\end{align}
to name two well-known examples, we retrieve the operator gain $\gamma$ and the input-feedforward passivity parameter $\rho$, respectively. The general dissipativity property specified by $(Q,S,R)$ will in the following also be referred to as $(Q,S,R)$-dissipativity. 

In the remainder of the paper, we make use of different equivalent conditions on dissipativity of an LTI system. The following standard result together with explanations and the proofs can be found, e.g., in \cite{Scherer2000,Kottenstette2014} and references therein. 
\begin{theorem}
Suppose that the system~\eqref{eq:sys} is controllable and let $s$ be a quadratic supply rate of the form~\eqref{eq:supply}. Then the following statements are equivalent.
\begin{itemize}
\item[a)] The system is $(Q,S,R)$-dissipative.
\item[b)] There exists a quadratic storage function $V(x) \coloneqq x^\top P x$ with $P = P^\top \succeq 0$ 
such that 
\begin{align*}
V(x_{k+1}) - V(x_k) \leq s(u_k,y_k)
\end{align*}
for all $k$ and all $(u,x,y)$ satisfying \eqref{eq:sys}.
\item[c)] There exists a matrix $P = P^\top \succeq 0$ such that 
\begin{align}
\label{eq:diss_lmi}
\begin{pmatrix} A^\top PA - P - \hat{Q} & A^\top PB - \hat{S} \\
(A^\top PB - \hat{S} )^\top & -\hat{R} + B^\top P B \end{pmatrix} \preceq 0
\end{align}
with  $\hat{Q} = C^\top QC$, $\hat{S} = C^\top S + C^\top QD$ and $\hat{R} = D^\top QD + (D^\top S + S^\top D) + R$.
\end{itemize}
\label{thm:diss_lmi}
\end{theorem}
\begin{remark}
The attentive reader might have noticed that, unlike in \cite{Scherer2000}, we require the storage function $V$ to be lower bounded, which yields the condition $P=P^\top \succeq 0$. 
Generally speaking, dissipativity as in \cite{Scherer2000} can be defined without requiring a lower bounded storage function.
However, a key motivation of inferring dissipativity properties from data (and hence a key motivation of the present paper) is to use such dissipativity properties in order to design controllers, e.g., for closed-loop stability.
In this case, it is meaningful to only consider lower bounded storage functions, similar to much of the related literature (cf.\ e.g.\ \cite{Kottenstette2014}).
The results in this paper can be directly extended to using data to verify ''cyclo-dissipativity'', compare~\cite{Willems2007}, in which case the storage function does not need to be bounded from below (i.e., $P\nsucceq0$ in Thm.~\ref{thm:diss_lmi}). Similarly, if a positive definite storage function is desired, one can simply substitute $P \succeq 0$ by $P \succ 0$ in Thm.~\ref{thm:diss_lmi}.
\end{remark}

Other approaches to determine dissipativity from data rely on an input-output formulation of dissipativity (e.g.~\cite{Maupong2017,Koch2020,Koch2020a,Wahlberg2010}). To put this into perspective, the following result shows that this input-output definition is equivalent to Def.~\ref{def:1}.
\begin{theorem}[\cite{Hill1980}]
\label{thm:Hill}
A system~\eqref{eq:sys} is dissipative w.r.t. the supply rate $s$ in \eqref{eq:supply} according to Def.~\ref{def:1} if and only if
\begin{align}
\sum_{k=0}^r s(u_k,y_k)
\geq0, \quad \forall r \geq 0,
\label{eq:diss}
\end{align}
for all trajectories $\{u_k,y_k\}_{k=0}^\infty$ of \eqref{eq:sys} with initial condition $x_0=0$. 
\end{theorem}
While this result shows the equivalence of the state-space definition of dissipativity and the input-output definition, in many other works where dissipativity is determined from data, dissipativity is only considered over a finite horizon. In these works (e.g.~\cite{Maupong2017,Koch2020,Wahlberg2010,Rojas2012,Romer2019c}), the condition~\eqref{eq:diss} is only verified over the horizon $r$ for $r\leq L$, which is also called $L$-dissipativity. 
Throughout this paper, we consider the classical definition of dissipativity as provided in Def.~\ref{def:1}.

Furthermore, Thm.~\ref{thm:Hill} also allows to infer dissipativity of systems which are not given in a minimal realizations by investigating dissipativity of a minimal realization with the same input-output behavior.
Whenever two systems have the same input-output behavior (i.e.~same span of input-output trajectories with zero initial condition), they satisfy the same condition~\eqref{eq:diss}. This insight will be especially important in Sec.~\ref{sec:output} and~\ref{sec:output_noise} when considering data-driven dissipativity from input-output data.

In the remainder of this paper, we use the equivalences stated in Thm.~\ref{thm:diss_lmi} and Thm.~\ref{thm:Hill} to verify or find dissipativity properties from data. 
We start in the following section by considering noise-free input and state trajectories.

\section{Data-driven dissipativity from input-state trajectories}
\label{sec:diss}
With the definitions and analysis in the previous section, we can directly state an equivalent formulation for dissipativity from noise-free input and state trajectories. The necessary and sufficient condition is a simple LMI that can be solved using standard solvers. 
\begin{theorem}[\cite{Koch2020a}]
\label{thm:1}
Given input and state trajectories $\{u_k\}_{k=0}^{N-1}$, $\{x_k\}_{k=0}^{N}$ of a controllable LTI system $G$ and the feasibility problem to find $P=P^\top\succeq0$ such that
\begin{align}
\begin{split}
&X_+^\top P X_+ - X^\top P X \\ 
&- \begin{pmatrix} U \\ CX+DU \end{pmatrix}^\top \begin{pmatrix} R & S^\top \\ S & Q \end{pmatrix} \begin{pmatrix} U \\ CX+DU \end{pmatrix}\preceq 0.
\end{split}
\label{eq:opt_allg}
\end{align}
\begin{enumerate}
\item If there exists a $P=P^\top\succeq0$ such that \eqref{eq:opt_allg} holds and, additionally, the rank condition \eqref{eq:rank} is satisfied, 
then $G$ is $(Q,S,R)$-dissipative.
\item If there exists no $P=P^\top\succeq0$ such that \eqref{eq:opt_allg} holds, then $G$ is not $(Q,S,R)$-dissipative.
\end{enumerate} 
\end{theorem}
\begin{proof} Substituting $X_+ = A X + B U$, the semidefinitness condition in \eqref{eq:opt_allg} can be equivalently written as 
\begin{align}
\begin{pmatrix} X \\ U \end{pmatrix}^\top 
\begin{pmatrix} A^\top PA - P - \hat{Q} & A^\top PB - \hat{S} \\
(A^\top PB - \hat{S} )^\top & -\hat{R} + B^\top P B \end{pmatrix} 
 \begin{pmatrix} X \\ U \end{pmatrix}  
\label{eq:pf_lmi}
\end{align}
with  $\hat{Q} = C^\top QC$, $\hat{S} = C^\top S + C^\top QD$ and $\hat{R} = D^\top QD + (D^\top S + S^\top D) + R$.
\begin{enumerate}
\item With \eqref{eq:rank}, the semidefiniteness condition \eqref{eq:pf_lmi} in turn implies that \eqref{eq:diss_lmi} holds, which implies dissipativity by Thm.~\ref{thm:diss_lmi}.
\item If problem~\eqref{eq:opt_allg} is infeasible, this directly implies that \eqref{eq:diss_lmi} is not negative semidefinite for any $P$, i.e. $G$ is not dissipative by Thm.~\ref{thm:diss_lmi}.
\end{enumerate}
\end{proof}

\begin{remark}
The condition \eqref{eq:rank} can easily be checked for the available data. With Lem.~\ref{lem:pe}, this rank condition can also be enforced by requiring or choosing the input $\{u_k\}_{k=0}^{N-1}$ to be persistently exciting of order $n+1$. Note that the latter condition requires a minimum length of the input trajectory, namely $N \geq (m + 1)n + m$. 
\end{remark}

The result stated in Thm.~\ref{thm:1} is conceptually similar to the approach in \cite{Waarde2020a} where methods for data-based system analysis (e.g., controllability, stability) are provided by verifying such properties for all systems which are consistent with the data. 
The data-based formulation of dissipativity given by Thm.~\ref{thm:1} is particularly simple and only requires solving a single semidefinite program. 
The proof relies on the fact that, if the matrix $\begin{pmatrix} X \\ U \end{pmatrix}$ has full row rank, then it spans all possible system trajectories. 
Multiplying~\eqref{eq:diss_lmi} from both sides by this matrix and exploiting the system dynamics $X_+=AX+BU$, we obtain the stated result.

In contrast to other input-output approaches (e.g.~\cite{Maupong2017,Koch2020,Wahlberg2010}), we exploited here the state-space definition of dissipativity which can be verified by looking at a difference viewpoint, i.e. looking at the difference at two time points (cf.\ Def.~\ref{def:1}). This yields the advantages, compared to many other data-driven dissipativity approaches, that rigorous guarantees on the infinite horizon as well as in the noisy case can be obtained, as will be discussed in the next section.  

\section{Dissipativity properties from noisy input-state trajectories}
\label{sec:noise}
While Sec.~\ref{sec:diss} provides a simple, computationally attractive condition to verify dissipativity properties of unknown systems, it assumes that exact measurements of input and state variables are available.
This assumption does rarely hold in practice.
Therefore, in this section, we extend the results to the case that the measured data are affected by noise.
More precisely, we consider in this section a variation of \eqref{eq:sys} that is disturbed by process noise of the form
\begin{align}
\begin{split}
x_{k+1} &= A x_k + B u_k + B_w w_k, \\
y_k &= C x_k + D u_k, 
\end{split}
\label{eq:sys_noise}
\end{align}
where $w_k \in \mathbb{R}^{m_w}$ denotes the noise and $B_w \in \mathbb{R}^{n \times m_w}$ is some known matrix describing the influence of the noise on the system dynamics. 

\begin{remark}
\label{rem:bw}
Note that including a known matrix $B_w$ into the analysis offers the possibility to include additional knowledge on the influence of the process noise on the system into the optimization problem. If no additional information on the effect of the noise on the different states is available, one can simply choose the identity matrix $B_w = I$.
\end{remark}

We denote the actual noise sequence which yields the available input-state trajectory $\{u_k\}_{k=0}^{N-1}$, $\{x_k\}_{k=0}^{N}$ by $\{\hat{w}_k\}_{k=0}^{N-1}$. While this noise sequence $\{\hat{w}_k\}_{k=0}^{N-1}$ is unknown, we assume that some information on the noise is available in form of a bound on the stacked matrix 
\begin{align*}
\hat{W} = \begin{pmatrix} \hat{w}_0 & \hat{w}_1 & \cdots & \hat{w}_{N-1} \end{pmatrix}
\end{align*}
as specified in the following assumption.
\begin{assumption} 
The matrix $\hat{W}$ denoting the stacked process noise $\{\hat{w}_k\}_{k=0}^{N-1}$ is an element of the set
\begin{align}
\mathcal{W} = \{ W \in \mathbb{R}^{m_w \times N} | \begin{pmatrix} W^\top \\ I \end{pmatrix}^\top \begin{pmatrix} Q_w & S_w \\ S_w^\top & R_w \end{pmatrix} \begin{pmatrix} W^\top \\ I \end{pmatrix} \succeq 0 \}
\label{eq:W}
\end{align}
with $Q_w \in \mathbb{R}^{N \times N}$, $S_w \in \mathbb{R}^{N \times m_w}$ and $R_w \in \mathbb{R}^{m_w \times m_w}$ with $Q_w \prec 0$. 
\end{assumption}
This quadratic bound on the noise matrix $\hat{W}$ is a flexible noise or disturbance description. Similar bounds on the noise were also used, for example, in \cite{Berberich2019c,Koch2020a,Waarde2020b,Berberich2020}. This quadratic matrix bound can incorporate bounds on sequences ($\|\hat{w}\|_2 \leq \bar{w}$) and bounds on separate components ($\|\hat{w}_k\|_2 \leq \bar{w}$ for all $k$), to name a few exemplary cases.

Due to the presence of noise, there generally exist multiple matrix pairs $(A_d, B_d)$ which are consistent with the data for some noise sequence $W \in \mathcal{W}$. The set of all such matrix pairs consistent with the input-state data and the noise bound is in the following denoted by 
\begin{align*}
\Sigma_{X,U} = \{ (A_{\text{d}},B_{\text{d}}) | X_+ = A_{\text{d}}X+B_{\text{d}}U+ B_w W, W \in \mathcal{W} \}.
\end{align*}
By assumption, this set includes the system of interest $(A,B)$ which generated the data.

To verify that a system~\eqref{eq:sys} is indeed $(Q,S,R)$-dissipative from noisy data, it is necessary to verify that \textit{all} systems that are consistent with the data are $(Q,S,R)$-dissipative.
In the language of~\cite{Waarde2020a}, we verify whether the data are \emph{informative} for dissipativity.
For this, we make use of an equivalent representation of the set $\Sigma_{X,U}$ provided in~\cite{Waarde2020b,Berberich2020}. This new equivalent representation of $\Sigma_{X,U}$ is a key step for retrieving a non-conservative condition on dissipativity and decreasing the conservatism with respect to the results in \cite{Koch2020a}.

\begin{lemma}
\label{lem:para}
It holds that
\begin{align}
\Sigma_{X,U} = \{ (A_{\text{d}},B_{\text{d}}) | 
\begin{pmatrix} A_d^\top \\ B_d^\top \\ I \end{pmatrix}^\top \begin{pmatrix} \bar{Q}_w & \bar{S}_w  \\ \bar{S}_w^\top & \bar{R}_w \end{pmatrix} \begin{pmatrix} A_d^\top \\ B_d^\top \\ I \end{pmatrix} \succeq 0
\}
\label{eq:AdBdT}
\end{align}
with 
\begin{align*}
\bar{Q}_w &= \begin{pmatrix} X \\ U \end{pmatrix} Q_w \begin{pmatrix} X \\ U \end{pmatrix}^\top, \\
\bar{S}_w &= -\begin{pmatrix} X \\ U \end{pmatrix} ( Q_w X_+^\top + S_w B_w^\top ), \\ 
\bar{R}_w &= X_+ Q_w X_+^\top + X_+ S_w B_w^\top  + B_w S_w^\top X_+^\top + B_w R_w B_w^\top.
\end{align*}
\end{lemma}
\begin{proof}
This statement follows from \cite[Lem.~4 and Rem.~2]{Waarde2020b}.
\end{proof}

The equivalent formulation of $\Sigma_{X,U}$ in Lem.~\ref{lem:para}, which is only based on data and the noise bound, allows us to rewrite the problem in a form such that we can directly apply robust analysis tools from the literature~\cite{Scherer2000}.
 While a similar idea was already exploited in \cite{Koch2020a} to verify and find dissipativity properties from input-state data, only a superset of $\Sigma_{X,U}$ could be considered, hence introducing conservatism. 
 On the contrary, we improve this result in the following by providing non-conservative conditions on dissipativity.

It follows from Lem.~\ref{lem:para} that the set of all LTI systems consistent with the data can be written as
\begin{align*}
x_{k+1} = \begin{pmatrix} A_d & B_d \end{pmatrix} \begin{pmatrix} x_k \\ u_k \end{pmatrix}
\end{align*}
for some $\begin{pmatrix} A_d & B_d \end{pmatrix} \in \Sigma_{X,U}$. 
We can equivalently reformulate this uncertain system as a linear fractional transformation (LFT) \cite{Zhou1996} of a nominal system with the 'uncertainty' $\begin{pmatrix} A_d & B_d \end{pmatrix}$, i.e., 
\begin{align}
\begin{split}
\begin{pmatrix}
x_{k+1} \\ \tilde{z}_k
\end{pmatrix}
= 
\begin{pmatrix}
0  & I \\
I  & 0
\end{pmatrix}
\begin{pmatrix}
\begin{pmatrix} x_k \\ u_k \end{pmatrix} \\ \tilde{w}_k
\end{pmatrix} \\
\tilde{w}_k = \begin{pmatrix} A_d & B_d \end{pmatrix} \tilde{z}_k,
\end{split}
\label{eq:lft}
\end{align}
with $\begin{pmatrix} A_d & B_d \end{pmatrix} \in \Sigma_{X,U}$.
This allows us to apply robust analysis results to guarantee dissipativity properties from noisy input-state trajectories, which is the main contribution in this section. 
To this end, we define
\begin{align}
\begin{pmatrix} \tilde{R} & \tilde{S}^\top \\ \tilde{S} & \tilde{Q} \end{pmatrix} = \begin{pmatrix} R & S^\top \\ S & Q \end{pmatrix}^{-1}
\label{eq:inv3}
\end{align}
assuming that the inverse exists.
\begin{theorem}
Let $\tilde{R} \succeq 0$. If there exists a matrix $P =P^\top \succ 0$, $\tau > 0$ such that \eqref{eq:rob_perf_lmi_dual} holds,
then~\eqref{eq:sys} is $(Q,S,R)$-dissipative for all matrices consistent with the data $(A_d, B_d) \in \Sigma_{X,U}$.
\begin{figure*}
\vspace{2pt}
\begin{align}\label{eq:rob_perf_lmi_dual}
\setstretch{1.25}
\mleft(
\begin{array}{ccc}
\left(I \quad 0 \right) & 0 & C^\top\\
0 &-I & 0\\ \hline
\left( 0 \quad I \right) & 0 & D^{\top} \\ 
0 & 0 &-I\\ \hline
I & 0 & 0\\
0 & I & 0\end{array}
\mright)^\top
\mleft(
\begin{array}{cc|cc|cc}
-P&0&0&0&0&0\\
0&P&0&0&0&0\\\hline
0&0&-\tilde{R}&{-}\tilde{S}^{\top}&0&0\\
0&0&-\tilde{S}&-\tilde{Q}&0&0\\\hline
0&0&0&0& -\tau \bar{Q}_w&{-}\tau \bar{S}_w\\
0&0&0&0&-\tau \bar{S}_w^\top& -\tau \bar{R}_w
\end{array}\mright)
\mleft(
\begin{array}{cccc}
\left(I \quad 0 \right) & 0 & C^\top\\
0 &-I & 0\\ \hline
\left( 0 \quad I \right) & 0 & D^{\top} \\ 
0 & 0 &-I\\ \hline
I & 0 & 0\\
0 & I & 0\end{array}
\mright)
\succ0
\end{align}
\noindent\makebox[\linewidth]{\rule{\textwidth}{0.4pt}}
\end{figure*}
\label{thm:noise}
\end{theorem}
\begin{proof}
By the full-block S-procedure \cite{Scherer2001} and using Lem.~\ref{lem:para}, \eqref{eq:rob_perf_lmi_dual} implies that 
\begin{align*}
\setstretch{1.2}
\setlength\arraycolsep{4pt}
\mleft(
\begin{array}{cccc}
\star & \star \\ 
\star & \star \\ \hline
\star & \star \\ 
\star & \star
\end{array}
\mright)^\top
\mleft(
\begin{array}{cc|cc}
-P&0&0&0\\
0&P&0&0\\\hline
0&0&-\tilde{R}&{-}\tilde{S}^{\top}\\
0&0&-\tilde{S}&-\tilde{Q}
\end{array}\mright)
 \mleft(
\begin{array}{cccc}
A_{\text{d}}^\top & C^\top \\ 
- I & 0 \\ \hline
B_{\text{d}}^\top & D^\top \\ 
0 & - I
\end{array}
\mright) \succ 0
\end{align*}
holds for all $(A_{\text{d}},B_{\text{d}}) \in \Sigma_{X,U}$.
Using the dualization lemma \cite{Scherer2000}, this in turn proves that 
\begin{align*}
\setstretch{1.2}
\setlength\arraycolsep{3.5pt}
\mleft(
\begin{array}{cccc}
\star & \star \\ 
\star & \star \\ \hline
\star & \star \\ 
\star & \star
\end{array}
\mright)^\top
\mleft(
\begin{array}{cc|cc}
-P^{-1}&0&0&0\\
0&P^{-1}&0&0\\\hline
0&0&-R&-S^{\top}\\
0&0&-S&-Q
\end{array}\mright)
 \mleft(
\begin{array}{cccc}
I & 0 \\
A_{\text{d}} & B_{\text{d}} \\ \hline
0 & I \\ 
C & D 
\end{array}
\mright) \prec 0
\end{align*}
holds for all $(A_{\text{d}},B_{\text{d}}) \in \Sigma_{X,U}$.
By Thm.~\ref{thm:diss_lmi}, this implies that \eqref{eq:sys} is $(Q,S,R)$-dissipative for all matrices consistent with the data $(A_d, B_d) \in \Sigma_{X,U}$, which concludes the proof.
\end{proof}

\begin{remark}
While for Thm.~\ref{thm:1}, $P\nsucceq0$ corresponds to cyclo-dissipativity, $P\succeq0$ (as stated) to $(Q,S,R)$-dissipativity and choosing $P\succ0$ together with a strict inequality in \eqref{eq:opt_allg} to strict dissipativity, we need to restrict our attention to $P \succ 0$ in Thm.~\ref{thm:noise} due to the dualization step.
\end{remark}

\begin{remark}
Thm.~\ref{thm:noise} provides a powerful tool to verifiy dissipativity properties from noisy input-state measurements, based on a simple LMI~\eqref{eq:rob_perf_lmi_dual}.
It verifies dissipativity for a tight description of the systems consistent with the data, i.e. in a non-conservative way.
Further reducing conservatism, e.g., by considering parameter-dependent storage functions, is an interesting issue for future research.
\end{remark}
Compared to the approach presented in~\cite{Koch2020a}, Thm.~\ref{thm:noise} provides a computationally less expensive approach for robust analysis of dissipativity from data, as the number of decision variables does not grow with the data length.

\begin{remark}
Since the definiteness condition in \eqref{eq:rob_perf_lmi_dual} is linear in the matrices $(Q,S,R)$, optimizing over specific dissipativity related parameters or the matrices $Q$,  $S$ or $R$ can be done via a simple SDP. For finding the operator gain $\gamma$, for example, choose $\tilde{R}=\frac{1}{\gamma^2} I$, $\tilde{S}=0$ and $\tilde{Q}=-I$ and minimize $-\frac{1}{\gamma^2}$ such that~\eqref{eq:rob_perf_lmi_dual} holds. 
\end{remark}

\begin{example}
\label{ex:1}
To illustrate the introduced approach, we apply it to a numerical example. We choose a randomly generated system of order
$n=5$ with two inputs and outputs $m=p=2$. We simulate a trajectory with $u_k$, $k=0,\dots,N$, uniformly sampled in $[-1,1]$ for different lengths $N$.
We sample the disturbance $\hat{w}_k$ uniformly from the ball $\|\hat{w}_k\|_2 \leq \bar{w}$ for all $k=1, \dots, N$, where $\bar{w} = 0.01$. This implies a bound on the measurement noise given by $\hat{W} \hat{W}^\top \preceq \bar{w}^2 N I$ for the respective data length $N$. The true operator gain of the randomly generated system is $\gamma_{\text{true}} = 11.44$. Applying the results from Thm.~\ref{thm:noise} by using~\eqref{eq:rob_perf_lmi_dual}, we retrieve an upper bound on the operator gain that is guaranteed for all systems consistent with the data. 
The resulting upper bounds on the operator gain for different data lengths $N$, are depicted in Fig.~\ref{fig:ex1}.
\begin{figure}[h]
\begin{tikzpicture}
\begin{axis}[%
width=0.35\textwidth,
height=0.2\textwidth,
at={(0cm,0cm)},
scale only axis,
xmin=7,
xmax=55,
xlabel={$N$},
ymin=0,
ymax=200,
ylabel={$\hat{\gamma}$},
axis background/.style={fill=white},
ymajorgrids,
]
\addplot [only marks,color=mycolor1,mark size=1.5pt,mark=*,mark options={solid},forget plot]
  table[row sep=crcr]{%
7 132.3\\
8 106.0\\
9 84.9\\
10 193.6\\
11 41.6\\
12 93.2\\
13 177.4\\
14 53.8\\
15 14.8\\
16 15.8\\
17 60.1\\
18 110.8\\
19 14.9\\
20 83.9\\
21 17.0\\
22 56.4\\
23 25.0\\
24 14.9\\
25 24.5\\
26 15.3\\
27 16.1\\
28 13.0\\
29 24.1\\
30 15.2\\
31 13.0\\
32 20.3\\
33 14.2\\
34 15.0\\
35 62.4\\
36 13.4\\
37 16.1\\
38 19.1\\
39 17.6\\
40 14.6\\
41 13.6\\
42 14.4\\
43 15.5\\
44 15.8\\
45 15.8\\
46 12.5\\
47 13.7\\
48 13.8\\
49 15.2\\
50 16.5\\
51 17.2\\
52 13.6\\
53 15.3\\
54 17.4\\
55 13.0\\  
56 13.4\\
57 15.3\\
58 25.6\\
59 13.3\\
60 39.4\\  
};

\addplot [color=orange,solid,forget plot,line width=1pt]
  table[row sep=crcr]{%
0	11.44\\
60 11.44\\
};
\end{axis}
\end{tikzpicture}%
\caption{Guaranteed upper bound on the operator gain from noisy input-state trajectories of different length $N$ for Ex.~\ref{ex:1}. 
}
\label{fig:ex1}
\end{figure}
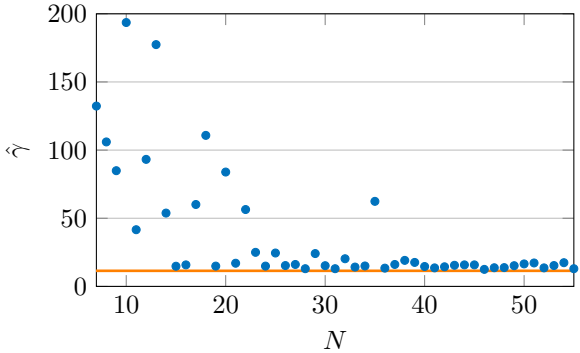

Next, we take the same example and increase the noise level $\bar{w}$ from $0.001$ to $0.02$ for $N=50$ data points each and again apply the result from Thm.~\ref{thm:noise}. The resulting upper bounds on the operator gain are depicted in Fig.~\ref{fig:ex1_noise}. 
\begin{figure}
\begin{tikzpicture}
\begin{axis}[%
width=0.35\textwidth,
height=0.2\textwidth,
at={(0cm,0cm)},
scale only axis,
xmin=0,
xmax=0.02,
xtick scale label code/.code={$\bar{w}$ $(10^{#1}$) \quad \quad \quad \quad \quad \quad \quad \quad \quad \quad \quad \quad \quad},
ymin=10,
ymax=58,
ylabel={$\hat{\gamma}$},
axis background/.style={fill=white},
ymajorgrids,
]
\addplot [only marks,color=mycolor1,mark size=1.5pt,mark=*,mark options={solid},forget plot]
  table[row sep=crcr]{%
0.001 11.56\\
0.002 11.69\\
0.003 11.83\\
0.004 11.98\\
0.005 12.15\\
0.006 12.34\\
0.007 12.54\\
0.008 12.78\\
0.009 13.05\\
0.010 13.37\\
0.011 13.76\\
0.012 14.22\\
0.013 14.81\\
0.014 15.57\\
0.015 16.59\\
0.016 18.05\\
0.017 20.27\\
0.018 24.07\\
0.019 31.85\\
0.020 55.86\\
};

\addplot [color=orange,solid,forget plot,line width=1pt]
  table[row sep=crcr]{%
0	11.44\\
0.02 11.44\\
};
\end{axis}
\end{tikzpicture}%
\caption{Guaranteed upper bounds on the operator gain of the system in Ex.~\ref{ex:1} from noisy input-state trajectories for increasing noise levels $\bar{w}$ and $N=50$ data points. }
\label{fig:ex1_noise}
\end{figure}
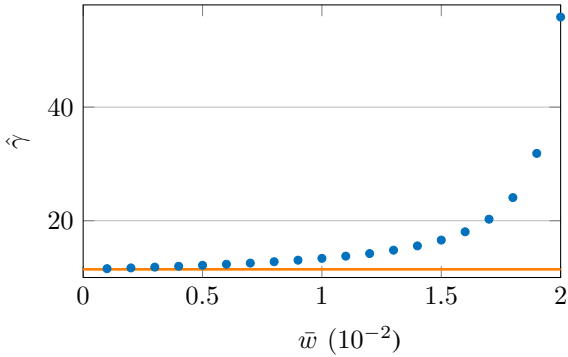

The results in both figures (Fig.~\ref{fig:ex1} and Fig.~\ref{fig:ex1_noise}) 
are generally well aligned with the theoretical guarantees. The computed  
$\hat{\gamma}$ is indeed always an upper bound on the true operator gain. Furthermore, it can be seen that for increasing noise bounds, the result becomes more conservative, as can be expected. More data points, i.e. longer trajectories, on the other hand, generally improve the result and make the bound tighter.
\end{example}

Thm.~\ref{thm:noise} can be seen as a counterpart to the data-driven controller design presented in \cite{Waarde2020b,Berberich2020}, focusing on data-driven system analysis instead. By applying the results of \cite{Berberich2020}, it is also possible to include prior model knowledge into the computation, if available.

As we discuss above, the result in Thm.~\ref{thm:noise} is powerful, being non-conservative and computationally simple. 
However, it also requires the availability of state measurements, which can be restrictive in practice, where often only input-output data are available.
In the next section, we extend the results of Sec.~\ref{sec:diss} to an input-output setting.
\section{Dissipativity from input-output trajectories}
\label{sec:output}
Instead of input and state measurements, we consider in this section the case where only an input-output sequence $\{u_k,y_k\}_{k=0}^{N-1}$ of~\eqref{eq:sys} is available, and we use this sequence to verify dissipativity properties.
We start by defining the lag of a system. 
\begin{definition}
The lag $\underline{l}$ of system~\eqref{eq:sys} is the smallest integer $l \in \mathbb{N}_{+}$ such that the observability matrix given by 
\begin{align*}
\mathcal{O}_l \coloneqq \begin{pmatrix} C \\ CA \\ \vdots \\ CA^{l-1}
\end{pmatrix}
\end{align*}
has rank $n$.
\label{def:lag}
\end{definition}
In this section, we use an extended state, based on $\underline{l}$ consecutive inputs and outputs, in order to verify dissipativity properties.
The following lemma shows that this is in principle possible since the corresponding stacked system has the same input-output behavior as~\eqref{eq:sys}.
\begin{lemma}
\label{lem:extended}
Let $l \geq \underline{l}$. Then there exists a system $\widetilde{G}$ with matrices $\widetilde{A},\widetilde{B},\widetilde{C},\widetilde{D}$ which can explain the data $\{u_k\}_{k=0}^{N-1}$, $\{y_k\}_{k=0}^{N-1}$, i.e., there exists $\xi_0$ such that for $k=0,\dots,N-1$,
\begin{align}\label{eq:sys_stacked}
\xi_{k+1} = \widetilde{A} \xi_k + \widetilde{B} u_k, \quad
y_k = \widetilde{C} \xi_k + \widetilde{D} u_k,
\end{align}
where the extended state is defined by
\begin{align*}
{\scriptsize
\xi_{k} = \begin{pmatrix} u_{k-l}^\top & u_{k-l+1}^\top & \cdots & u_{k-1}^\top & y_{k-l}^\top & y_{k-l+1}^\top & \cdots & y_{k-1}^\top \end{pmatrix}^\top.}
\end{align*} 
\end{lemma}

While Lem.~\ref{lem:extended} is a well-known fact, we nevertheless add a proof in the appendix for completeness and to provide some intuition.

The converse of Lem.~\ref{lem:extended} follows trivially from its proof and the constructed extended system in \eqref{eq:extended_sys}: All input-output trajectories of the extended system~\eqref{eq:sys_stacked} with zero initial condition $\xi_0 = 0$ (or $\xi_0 \in \mathcal{X}_\xi$, where $\mathcal{X}_\xi$ denotes the set of reachable states) are also input-output trajectories of the system~\eqref{eq:sys}.
We can hence conclude that 
the extended system~\eqref{eq:sys_stacked} has the same input-output behavior as \eqref{eq:sys} if the initial condition $\xi_0$ is restricted to the set of reachable states.  
This implies that both systems have the same input-output behavior for zero initial condition $\xi_0 = 0, x_0 = 0$. 
Together with Thm.~\ref{thm:Hill}, this in turn implies that dissipativity of the extended system~\eqref{eq:sys_stacked} (if the initial condition is reachable) 
is equivalent to dissipativity of \eqref{eq:sys}. 
Hence, using Lem.~\ref{lem:extended}, we can reduce the problem of verifying dissipativity from input-output trajectories to the problem of verifying dissipativity from input-state trajectories of the potentially non-minimal system~\eqref{eq:sys_stacked}.

\begin{remark}
In a purely data-driven setup, knowledge on the lag $\underline{l}$ is often not available. However, as shown above, it is sufficient for the purpose of this section to have an upper bound on $l$ or even an upper bound an $n$ since $\underline{l} \leq n$. 
\end{remark}
Similar to the results of this section,~\cite{Persis2019} uses such an extended state to design data-driven controllers but, instead of the lag $\underline{l}$, the system order $n$ is used. For MIMO systems, this can result in significantly larger state dimensions. Furthermore, \cite{Persis2019} assumes controllability of the extended system, which is generally not the case, unless $l=n$ and SISO systems are considered.

Before stating our main result of this section, we recall the Fundamental Lemma introduced in \cite{Willems05}: 
\begin{lemma}[Fundamental Lemma]
\label{lem:fundamental}
Suppose $\{u_k,y_k\}_{k=0}^{N-1}$ is a trajectory of a controllable LTI system $G$, where $u$ is persistently exciting of order $l+n$.
Then, $\{\bar{u}_k,\bar{y}_k\}_{k=0}^{l-1}$ is a trajectory of $G$ if and only if there exists $\alpha\in\mathbb{R}^{N-l+1}$ such that
\begin{align*}
\begin{bmatrix}H_l(u)\\H_l(y)\end{bmatrix}\alpha
=\begin{bmatrix}\bar{u}\\\bar{y}\end{bmatrix}.
\end{align*}
\end{lemma}

Note that the Fundamental Lemma requires controllability. However, since system~\eqref{eq:sys} that generated the data is assumed to be minimal, we can apply Lem.~\ref{lem:fundamental} to describe all input-output trajectories of~\eqref{eq:sys} and hence also of the extended system for $\xi_0 \in \mathcal{X}_\xi$ applying Lem.~\ref{lem:extended}.
Using Lem.~\ref{lem:extended} and Lem.~\ref{lem:fundamental}, we can hence determine dissipativity properties from input-output trajectories.
For this, 
we collect the extended state data analogously to Sec.~\ref{sec:diss} in the following form
\begin{align*}
\Xi &\coloneqq \begin{pmatrix} \xi_l & \xi_{l+1} & \cdots & \xi_{N-1} \end{pmatrix}, \\
\Xi_+ &\coloneqq \begin{pmatrix} \xi_{l+1} & \xi_{l+2} & \cdots & \xi_{N} \end{pmatrix}, \\ 
Y_\Xi &\coloneqq \begin{pmatrix} y_{l} & y_{l+1} & \cdots & y_{N-1} \end{pmatrix}, \\
U_\Xi &\coloneqq \begin{pmatrix} u_{l} & u_{l+1} & \cdots & u_{N-1} \end{pmatrix},
\end{align*}
which directly leads us to the main result of this section.

\begin{theorem}
\label{thm:3}
Given an input-output trajectory $\{u_k, y_k\}_{k=0}^{N-1}$ of a controllable LTI system $G$ of the form~\eqref{eq:sys} with lag $\underline{l}$. Let $l \geq \underline{l}$ and consider the feasibility problem to find $P=P^\top\succeq0$ such that
\begin{align}
\begin{split}
&\Xi_+^\top P \Xi_+ - \Xi^\top P \Xi \\
&- Y_\Xi^\top Q Y_\Xi  
- Y_\Xi^\top SU_\Xi - (SU_\Xi)^\top Y_\Xi - U_\Xi^\top R U_\Xi \preceq 0.
\end{split}
\label{eq:opt_output}
\end{align}
\begin{enumerate}
\item If there exists a $P=P^\top\succeq0$ such that \eqref{eq:opt_output} holds  and additionally $\{u_k\}_{k=0}^{N-1}$ is persistently exciting of order $n+l+1$, then $G$ is $(Q,S,R)$-dissipative.
\item If there exists no $P=P^\top\succeq0$ such that \eqref{eq:opt_output} holds, then $G$ is not $(Q,S,R)$-dissipative.
\end{enumerate} 
\end{theorem}
\begin{proof} 
1) First, we notice that the 
data matrix $\Xi$ can be written as 
\begin{align*}
\Xi = \begin{pmatrix} H_l (\{u_k\}_{k=0}^{N-2}) \\ H_l (\{y_k\}_{k=0}^{N-2}) \end{pmatrix} = \begin{pmatrix}
u_0 & u_{1} & \cdots & u_{N-l-1} \\
u_1 & u_{2} & \cdots & u_{N-l} \\ 
\vdots & & & \vdots \\
u_{l-1} & u_{l} & \cdots & u_{N-2} \\
y_{0} & y_{1} & \cdots & y_{N-l-1} \\
y_{1} & y_{2} & \cdots & y_{N-l} \\
\vdots & & & \vdots \\
y_{l-1} & y_{l} & \cdots & y_{N-2} \\
\end{pmatrix}.
\end{align*}
Since there exists a controllable realization (of order $n$) with the same input-output behavior as the extended system, 
the Fundamental Lemma implies that the image of 
$\Xi$ spans the whole reachable state space of the extended system $\mathcal{X}_\xi$. More specifically, if $\{u_k\}_{k=0}^{N-1}$ is persistently exciting of order $n+l$, then Lem.~\ref{lem:fundamental} guarantees that the columns in $\Xi$ span all possible input-output trajectories of the system $G$, and hence the whole reachable state space of the extended system~\eqref{eq:sys_stacked}.
If $\{u_k\}_{k=0}^{N-1}$ is persistently exciting of order $n+l+1$, 
then it additionally holds that $\begin{pmatrix}\Xi\\U_\Xi\end{pmatrix}$ spans the space of all input-state trajectories of~\eqref{eq:sys_stacked}.

Using $\Xi_+=\tilde{A}\Xi+\tilde{B}U_\Xi$ and rearranging~\eqref{eq:opt_output}, we obtain
\begin{align}
\begin{pmatrix} \Xi \\ U_\Xi \end{pmatrix}^\top 
\begin{pmatrix} \widetilde{A}^\top P \widetilde{A} - P - \hat{Q} & \widetilde{A}^\top P \widetilde{B} - \hat{S} \\
(\widetilde{A}^\top P \widetilde{B} - \hat{S} )^\top & -\hat{R} + \widetilde{B}^\top P \widetilde{B} \end{pmatrix} 
 \begin{pmatrix} \Xi \\ U_\Xi \end{pmatrix} \preceq 0,
\label{eq:proof11}
\end{align}
with $\hat{Q}$, $\hat{S}$, $\hat{R}$ similar as in \eqref{eq:diss_lmi}.
Since $\begin{pmatrix}\Xi\\U_\Xi\end{pmatrix}$ spans the space of all input-state trajectories, this implies
\begin{align}
\begin{pmatrix} \xi_k \\ u_k \end{pmatrix}^\top 
\begin{pmatrix} \widetilde{A}^\top P \widetilde{A} - P - \hat{Q} & \widetilde{A}^\top P \widetilde{B} - \hat{S} \\
(\widetilde{A}^\top P \widetilde{B} - \hat{S} )^\top & -\hat{R} + \widetilde{B}^\top P \widetilde{B} \end{pmatrix} 
 \begin{pmatrix} \xi_k \\ u_k \end{pmatrix} \leq 0
\label{eq:proof1}
\end{align}
for all $k$ and all trajectories $(u,\xi)$ of the extended system~\eqref{eq:sys_stacked} (with $\xi_0 \in \mathcal{X}_\xi$). 
Hence, there exists a quadratically lower bounded storage function for the extended system~\eqref{eq:sys_stacked} satisfying the dissipation inequality. This implies that the system~\eqref{eq:sys_stacked} is $(Q,S,R)$-dissipative which in turn, using Thm.~\ref{thm:Hill} and Lem.~\ref{lem:extended}, implies that the system~\eqref{eq:sys} is $(Q,S,R)$-dissipative.

2) We prove this direction via contraposition. If system~\eqref{eq:sys} is $(Q,S,R)$-dissipative, then, according to Thm.~\ref{thm:diss_lmi}, there exists a quadratic storage function $V(x_k) = x_k^\top P^\prime x_k$ such that
\begin{align*}
x_{k+1}^\top P^\prime x_{k+1} - x_k^\top P^\prime x_k \leq s(u_k, y_k)
\end{align*}
holds for all $k$ and all $(u,x,y)$ satisfying~\eqref{eq:sys}. From the proof of Lem.~\ref{lem:extended}, we know that there exists a transformation matrix $T$ such that $x_{k} = T \xi_k$ holds for all reachable states $\xi_k$ and all $k$. Hence, the matrix $P = T^\top P^\prime T\succeq0$ 
satisfies \eqref{eq:proof1} for all $k$ and all $(u,\xi)$ of the extended system~\eqref{eq:sys_stacked}. 
Using the Fundamental Lemma~\cite{Willems05}, this implies that \eqref{eq:proof11} holds and thus there exists a $P\succeq0$ such that \eqref{eq:opt_output} holds.
\end{proof}

Thm.~\ref{thm:3} provides an equivalent formulation of dissipativity based on input-output data.
The result itself and its proof are conceptually similar to the state measurements case in Thm.~\ref{thm:1}.
A key challenge is that, in contrast to the matrix $\begin{pmatrix}X\\U\end{pmatrix}$, the matrix $\begin{pmatrix}\Xi\\U_\Xi\end{pmatrix}$ does usually not have full row rank, even if the input is persistently exciting, since the system~\eqref{eq:sys_stacked} is usually not controllable.
However, the Fundamental Lemma implies that, assuming the input to be persistently exciting of order $l+n$, the matrix $\Xi$ spans the space of all state trajectories of the extended system~\eqref{eq:sys_stacked}.
Under the stronger assumption of persistence of excitation of order $l+n+1$, which we assume in Thm.~\ref{thm:3}, it even holds that $\begin{pmatrix}\Xi\\U_\Xi\end{pmatrix}$ spans the space of all input-state trajectories of~\eqref{eq:sys_stacked}.
Using this fact, it is then straightforward to derive~\eqref{eq:opt_output}, which provides an equivalent data-driven characterization of dissipativity.

\begin{remark}
In Thm.~\ref{thm:3}, a necessary and sufficient condition for $(Q,S,R)$-dissipativity is given, which can be extended to cyclo-dissipativity for $P \nsucceq 0$. However, for $l > \underline{l}$ it is generally difficult to verify strict dissipativity (with $P \succ 0$ and a strict definiteness condition in \eqref{eq:opt_output}), since $\begin{pmatrix}\Xi\\U_\Xi\end{pmatrix}$ does usually not have full row rank in these cases. Finding conditions for strict dissipativity from input-output data with $l >\underline{l}$ is therefore an interesting issue for future research.
\end{remark}

\section{Dissipativity from noisy input-output trajectories}
\label{sec:output_noise}
In this section, we extend the results of Sec.~\ref{sec:output} to the case of noisy input-output data. From an input-output viewpoint, we consider the system~\eqref{eq:sys} in the difference operator form
\begin{align}
\begin{split}
 y_k = &-a_{l} y_{k-1} - \dots - a_2 y_{k-l+1} - a_1 y_{k-l} \\
&+ d u_{k} + b_{l} u_{k-1} + \dots + b_2 u_{k-l+1} + b_1 u_{k-l},
\end{split}
\label{eq:sys_diff}
\end{align}
with $a_i \in \mathbb{R}^{p \times p}$, $b_i \in \mathbb{R}^{p \times m}$, $i=1, \dots, l$, and $l$ is an upper bound on the lag $l \geq \underline{l}$ (compare Sec.~\ref{sec:output}, Def.~\ref{def:lag} \& Lem.~\ref{lem:extended}). 

Instead of having exact measurements of the output, we assume that the input-output behavior is 
corrupted by process noise of the form
\begin{align}
\begin{split}
 y_k = &-a_{l} y_{k-1} - \dots - a_2 y_{k-l+1} - a_1 y_{k-l} \\
&+ d u_{k} + b_{l} u_{k-1} + \dots + b_2 u_{k-l+1} + b_1 u_{k-l} + b_v v_k,
\end{split}
\label{eq:sys_diff_with_noise}
\end{align}
where $v_k \in \mathbb{R}^{m_v}$ denotes the noise and, as before,  the choice of $b_v \in \mathbb{R}^{p \times m_v}$ 
 can be used to include prior knowledge on the influence of the noise (cf.~Rem.~\ref{rem:bw}). The noisy input-output behavior in~\eqref{eq:sys_diff_with_noise}
 can also be represented in state-space via \eqref{eq:extended_sys_bw}.
\begin{figure*}
\vspace{2pt}
\begin{align}
\begin{split}
\label{eq:extended_sys_bw}
\begin{pmatrix} u_{k-l+1} \\ \vdots \\ u_{k-1} \\ u_{k} \\ y_{k-l+1} \\ \vdots \\ y_{k-1} \\ y_{k} \end{pmatrix} = \begin{pmatrix} 
0 & I & \dots & 0 & 0 & 0 & \dots & 0  \\
\vdots & \ddots & \ddots&\ddots& \vdots & & \ddots&\vdots \\
0 & 0 & \dots & I & 0& 0&  \dots & 0 \\
0 & 0 & \dots & 0 & 0 & 0&  \dots & 0 \\
0 & 0 & \dots & 0 & 0& I& \dots &0 \\
\vdots & \ddots & \ddots&\vdots&\vdots & & \ddots&\vdots \\
0 & 0 & \dots & 0 & 0& 0&  \dots & I \\
b_1 & b_2 & \dots & b_l & -a_1 & -a_2 & \dots &  -a_l 
\end{pmatrix}  
\begin{pmatrix} u_{k-l} \\ u_{k-l+1} \\ \vdots \\ u_{k-1} \\ y_{k-l} \\ y_{k-l+1} \\ \vdots \\ y_{k-1}  \end{pmatrix} + 
\begin{pmatrix} 
0 \\ \vdots \\ 0 \\ I \\ 0 \\ \vdots \\ 0 \\ D  \end{pmatrix} 
u_{k} + \begin{pmatrix} 0 \\ \vdots \\ 0 \\ 0 \\ 0 \\ \vdots \\ 0 \\  b_v \end{pmatrix} v_{k} 
\end{split} 
\end{align}
\end{figure*}
Since only the last block row in \eqref{eq:extended_sys_bw} is uncertain, 
we introduce the notation
\begin{align}
\begin{split}
\label{eq:extended_sys2}
\xi_{k+1} &= 
\begin{pmatrix} 
\widetilde{A}_1 \\
\widetilde{A}_2 \end{pmatrix} 
\xi_k + \begin{pmatrix} \widetilde{B}_1 \\ \widetilde{D} \end{pmatrix} u_{k} + \begin{pmatrix} 0 \\  b_v \end{pmatrix} v_{k}, \\
y_k &= \widetilde{A}_2 \xi_k + \widetilde{D} u_k + b_v v_{k},
\end{split} 
\end{align}
where $\tilde{A}_1 \in \mathbb{R}^{((p+m)l-p) \times (p+m)l}$ and $\widetilde{B}_1 \in \mathbb{R}^{((p+m)l-p) \times m}$ are known (cf.~\eqref{eq:extended_sys_bw}), and $\tilde{A}_2 \in \mathbb{R}^{p \times (p+m)l}$ and $\widetilde{D} \in \mathbb{R}^{p \times m}$ are unknown. 

With the input-output viewpoint~\eqref{eq:sys_diff} and the extended state system representation~\eqref{eq:extended_sys_bw}, we can follow a similar approach as in Sec.~\ref{sec:noise} to derive dissipativity conditions based on noisy input-output data.
Similar as in the input-state case, we denote the actual noise sequence which yields the available input-output trajectory $\{u_k\}_{k=0}^{N-1}$, $\{y_k\}_{k=0}^{N-1}$ by $\{\hat{v}_k\}_{k=0}^{N-1}$. While the exact noise instance is unknown, we assume that we have information on a bound on the stacked matrix
\begin{align}
\hat{V} = \begin{pmatrix} \hat{v}_{l} & \hat{v}_{l+1}  & \dots & \hat{v}_{N-1} \end{pmatrix}
\label{eq:vhat}
\end{align}
as specified in the following assumption.
\begin{assumption}
\label{as:V}
The matrix $\hat{V}$ in \eqref{eq:vhat} is an element of the set
\begin{align*}
\mathcal{V} = \{ V \in \mathbb{R}^{p \times (N-l)} | 
\begin{pmatrix} V^\top \\ I \end{pmatrix}^\top \begin{pmatrix} Q_v & S_v \\ S_v^\top & R_v \end{pmatrix} \begin{pmatrix} V^\top \\ I \end{pmatrix} \succeq 0 \},
\end{align*}
where $Q_v \in \mathbb{R}^{(N-l) \times (N-l)}$, $S_v \in \mathbb{R}^{(N-l) \times m_v}$ and $R_v \in \mathbb{R}^{m_v \times m_v}$ with $Q_v \prec 0$. 
\end{assumption}
Due to the presence of noise, there generally exist multiple matrix pairs $(\widetilde{A}_2, \widetilde{D})$ which are consistent with the data for some noise sequence $V \in \mathcal{V}$. We denote the set of all such matrix pairs consistent with the input-output data and the noise bound by 
\begin{align*}
\Sigma_{U,Y} = \{ (\widetilde{A}_{2,\text{d}}, \widetilde{D}_{\text{d}}) | Y_\Xi = \widetilde{A}_{2,\text{d}} \Xi + \widetilde{D}_{\text{d}} U_\Xi + b_v V, V \in \mathcal{V} \}.
\end{align*}
Along the lines of Sec.~\ref{sec:noise}, this leads to an equivalent formulation for the set $\Sigma_{U,Y}$.
\begin{lemma}
\label{lem:para2}
It holds that
\begin{align}
\Sigma_{U,Y} = \{ (\widetilde{A}_{2,\text{d}}, \widetilde{D}_{\text{d}}) | 
\begin{pmatrix} \widetilde{A}_{2,\text{d}}^\top \\ \widetilde{D}_{\text{d}}^\top \\ I \end{pmatrix}^\top \begin{pmatrix} \bar{Q}_v & \bar{S}_v  \\ \bar{S}_v^\top & \bar{R}_v \end{pmatrix} \begin{pmatrix} \widetilde{A}_{2,\text{d}}^\top \\ \widetilde{D}_{\text{d}}^\top \\ I \end{pmatrix} \succeq 0
\}
\label{eq:At2Dt}
\end{align}
with 
\begin{align*}
\bar{Q}_{v} &= \begin{pmatrix} \Xi \\ U_\Xi\end{pmatrix} Q_v \begin{pmatrix} \Xi \\ U_\Xi\end{pmatrix}^\top, \\
\bar{S}_{v} &= -\begin{pmatrix} \Xi \\ U_\Xi\end{pmatrix} \left( Q_v Y_\Xi^\top + S_v b_v^\top \right), \\
\bar{R}_{v} &= Y_\Xi Q_v Y_\Xi^\top + Y_\Xi S_v b_v^\top + b_v S_v^\top Y_\Xi^\top + b_v R_v b_v^\top.
\end{align*}
\end{lemma}
\begin{proof}
With
\begin{align*}
\begin{pmatrix} V^\top \\ I \end{pmatrix} = \begin{pmatrix}- \Xi^\top & -U_\Xi^\top &  Y_\Xi^\top \\ 0 & 0 & I \end{pmatrix} \begin{pmatrix} \widetilde{A}_{2,\text{d}}^\top \\ \widetilde{D}_{\text{d}}^\top  \\ I \end{pmatrix},
\end{align*}
this lemma can be proven similar to the proof of \cite[Lem.~4 and Rem.~2]{Waarde2020b}.
\end{proof}

With the quadratic bound on the unknown matrices $( \widetilde{A}_{2, \text{d}} , \widetilde{D}_{\text{d}})$ from an input-output trajectory, we can again reformulate the uncertain system as a linear fractional transformation (LFT) of a nominal system 
\begin{align}
\begin{split}
\begin{pmatrix}
\xi_{k+1} \\ \tilde{z}_k
\end{pmatrix}
= 
\begin{pmatrix}
\begin{pmatrix}
\widetilde{A}_1 & \widetilde{B}_1 & 0 \\
0 & 0  & I 
\end{pmatrix}\\
\begin{pmatrix}
\phantom{_1}I\phantom{_1} & \phantom{_1}0\phantom{_1} & 0 \\
0 & I & 0 
\end{pmatrix}
\end{pmatrix}
\begin{pmatrix}
\begin{pmatrix} \xi_k \\ u_k \end{pmatrix} \\ \tilde{v}_k
\end{pmatrix} \\
\tilde{v}_k = \begin{pmatrix} \widetilde{A}_{2, \text{d}} & \widetilde{D}_{\text{d}} \end{pmatrix} \tilde{z}_k,
\end{split}
\label{eq:lft2}
\end{align}
with $( \widetilde{A}_{2, \text{d}} , \widetilde{D}_{\text{d}}) \in \Sigma_{U,Y}$. This allows us to again apply robust analysis results to guarantee dissipativity properties from noisy input-output trajectories. For this, 
we assume again that the inverse \eqref{eq:inv3} exists.
\begin{theorem}\label{thm:noise_io}
Let $\tilde{R} \preceq 0$. If there exists a matrix $P =P^\top \succ 0$, $\tau > 0$ such that \eqref{eq:rob_io_dual} holds,
then~\eqref{eq:sys_diff} is $(Q,S,R)$-dissipative for all matrices consistent with the data $( \widetilde{A}_{2, \text{d}} , \widetilde{D}_{\text{d}}) \in \Sigma_{U,Y}$.
\end{theorem}
\begin{proof}
With the result from Lem.~\ref{lem:para2}, the proof follows along the arguments of the proof of Thm.~\ref{thm:noise}.
\end{proof}

\begin{figure*}
\vspace{2pt}
\begin{align}\label{eq:rob_io_dual}
\setstretch{1.25}
\mleft(
\begin{array}{ccc}
\begin{pmatrix}\widetilde{A}_1^\top & 0\end{pmatrix} & 0 & \begin{pmatrix}I & 0\end{pmatrix} \\ 
-I & 0 & 0 \\
\begin{pmatrix} \widetilde{B}_1^\top & 0 \end{pmatrix} & 0 & \begin{pmatrix} 0 & I \end{pmatrix} \\ 
0 & -I & 0 \\
0 & 0 & I \\ 
\begin{pmatrix} 0&I \end{pmatrix} & I & 0 \\
\end{array}
\mright)^\top
\mleft(
\begin{array}{cc|cc|cc}
-P&0&0&0&0&0\\
0&P&0&0&0&0\\\hline
0&0&-\tilde{R}&-\tilde{S}^{\top}&0&0\\
0&0&-\tilde{S}&-\tilde{Q}&0&0\\\hline
0&0&0&0& -\tau \bar{Q}_{v}&-\tau \bar{S}_{v} \\
0&0&0&0&-\tau \bar{S}_{v}^\top & -\tau \bar{R}_{v}
\end{array}\mright)
\mleft(
\begin{array}{ccc}
\begin{pmatrix}\widetilde{A}_1^\top & 0\end{pmatrix} & 0 & \begin{pmatrix}I & 0\end{pmatrix} \\ 
-I & 0 & 0 \\
\begin{pmatrix} \widetilde{B}_1^\top & 0 \end{pmatrix} & 0 & \begin{pmatrix} 0 & I \end{pmatrix} \\ 
0 & -I & 0 \\
0 & 0 & I \\ 
\begin{pmatrix} 0&I \end{pmatrix} & I & 0 \\
\end{array}
\mright)
\succ0
\end{align}
\noindent\makebox[\linewidth]{\rule{\textwidth}{0.4pt}}
\end{figure*}

\begin{remark}
Note that from the proof of Lem.~\ref{lem:extended} (and the extended system as described in~\eqref{eq:extended_sys}), one can see that the matrices $a_i$, $b_i$, $i=1, \dots, l$ \eqref{eq:sys_diff} are uniquely defined if the left inverse of $\mathcal{O}_l$ is unique and hence if
$l = \underline{l}$ and $p \underline{l} = n$. 
Empirical evaluations showed that in these cases, the presented approach based on the condition in \eqref{eq:rob_io_dual} worked very well in numerical examples, while for overapproximations of the lag no reasonable upper bounds on the respective dissipativity properties could be found. Improving the approach in these cases is part of future work.
\end{remark}
\begin{example}
\label{ex:2}
We illustrate the introduced approach with a numerical example. We consider a randomly generated system of order $n=4$ with two inputs and outputs $m=p=2$. The system has an operator gain of $\gamma_{\text{true}} = 3.73$. We assume knowledge of the lag $\underline{l}=2$ and we simulate the trajectory with $u_k$, $k=0,\dots,N$ uniformly sampled in $[-1,1]$, for different lengths $N$. We sample the noise $\hat{v}_k$ uniformly from the ball $\|\hat{v}_k\|_2 \leq \bar{v}$ for all $k=0,\dots,N$ with $\bar{v} = 0.01$, which implies a bound on the measurement noise of $\hat{V}^\top \hat{V} \preceq \bar{v}^2 (N-\underline{l}) I$ (cf. Asm.~\ref{as:V}). We then apply the result of Thm.~\ref{thm:noise_io} and solve an SDP for finding the minimal $\gamma$ such that \eqref{eq:rob_io_dual} holds with $\Pi_\gamma$ in \eqref{eq:pi}. The resulting upper bound on the operator gain, which is guaranteed for all systems that are consistent with the data, is depicted in Fig.~\ref{fig:ex2} for different data lengths.

\begin{figure}[ht]
\definecolor{mycolor1}{rgb}{0.00000,0.44700,0.74100}%
\begin{tikzpicture}

\begin{axis}[%
width=0.35\textwidth,
height=0.2\textwidth,
at={(0cm,0cm)},
scale only axis,
xmin=15,
xmax=40,
xlabel={$N$},
ymin=3,
ymax=5,
ylabel={$\hat{\gamma}$},
axis background/.style={fill=white},
ymajorgrids,
]
\addplot [only marks,color=mycolor1,mark size=1.5pt,mark=*,mark options={solid},forget plot]
  table[row sep=crcr]{%
18	4.662761283325202\\
19   3.766959984817681\\
20   3.709745733237886\\
21   3.728304895613251\\
22   3.744160696200605\\
23   3.748417531335434\\
24   3.767320685494549\\
25   3.788802559149664\\
26   3.703386490519571\\
27   3.705129877124437\\
28   3.717208340067915\\
29   3.730358885578753\\
30   3.742347173353014\\
31   3.756354382971415\\
32   3.717138203825190\\
33   3.700445778056766\\
34   3.687731079197144\\
35   3.691157735627397\\
36   3.697965960903294\\
37   3.684441708723422\\
38   3.667987726870205\\
39   3.675752810886129\\
40   3.682385695555215\\
41   3.680937883116116\\
42   3.686669016093030\\
43   3.691162745359019\\
44   3.696543197122311\\
45   3.697184442770078\\
46   3.680662529178669\\
47   3.680329048911235\\
48   3.678377910256674\\
49   3.663858195624788\\
50   3.669080963117746\\
51   3.672421339356717\\
52   3.676945127012445\\
53   3.676714679357779\\
54   3.682445812334693\\
55   3.688036672826158\\
56   3.686278257026423\\
57   3.690882200388122\\
58   3.694584392057649\\
59   3.697807755102214\\
60   3.702930604941861\\
};
\addplot [only marks,color=red,mark size=4pt,mark=x,mark options={solid},forget plot]
  table[row sep=crcr]{%
15 3\\
16 3\\
17 3\\
};

\addplot [color=orange,solid,forget plot,line width=1pt]
  table[row sep=crcr]{%
0	3.361973981762519\\
60 3.361973981762519\\
};
\end{axis}
\end{tikzpicture}%
\caption{Guaranteed upper bound on the operator gain from noisy input-output trajectories (\textcolor{mycolor1}{$\bullet$}) 
of different length $N$ for Ex.~\ref{ex:2}. The red crosses (\textcolor{red}{$\times$}) indicate that no upper bound could be found.
}
\label{fig:ex2}
\end{figure}
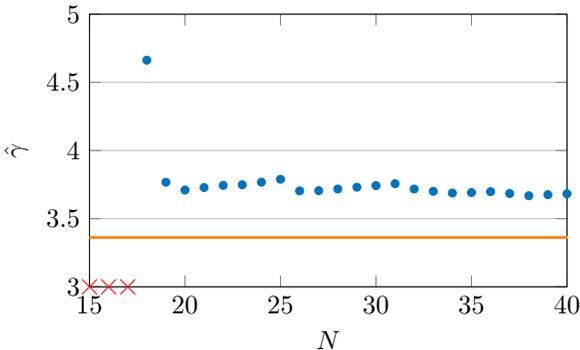

Next, we take the same example and increase the noise level $\bar{v}$ from $0.001$ to $0.022$ for $N=50$ data points each and again apply the result from Thm.~\ref{thm:noise_io}. The resulting upper bounds on the operator gain are depicted in Fig.~\ref{fig:ex2_noise}. 
\begin{figure}
\begin{tikzpicture}
\begin{axis}[%
width=0.35\textwidth,
height=0.2\textwidth,
at={(0cm,0cm)},
scale only axis,
xmin=0,
xmax=0.022,
xtick scale label code/.code={$\bar{v}$ $(10^{#1}$) \quad \quad \quad \quad \quad \quad \quad \quad \quad \quad \quad \quad},
ymin=3,
ymax=8,
ylabel={$\hat{\gamma}$},
axis background/.style={fill=white},
ymajorgrids,
]
\addplot [only marks,color=mycolor1,mark size=1.5pt,mark=*,mark options={solid},forget plot]
  table[row sep=crcr]{%
0.001 3.3808060172651\\
0.002 3.4001923102446\\
0.003 3.4218283138262\\
0.004 3.4455705634803\\
0.005 3.4718316851647\\
0.006 3.5011486346236\\
0.007 3.5342028438733\\
0.008 3.5719010743377\\
0.009 3.6154657005328\\
0.01 3.6665749825525\\
0.011 3.7276335915757\\
0.012 3.8026679150828\\
0.013 3.8956449610222\\
0.014 4.0168604273766\\
0.015 4.1811892032281\\
0.016 4.4176917235354\\
0.017 4.7896065205199\\
0.018 5.4634445600252\\
0.019 7.0650189573566\\
};
\addplot [only marks,color=red,mark size=4pt,mark=x,mark options={solid},forget plot]
  table[row sep=crcr]{%
0.02 3\\
0.021 3\\
0.022 3\\
};
\addplot [color=orange,solid,forget plot,line width=1pt]
  table[row sep=crcr]{%
0 3.361973981762519\\
0.025 3.361973981762519\\
};
\end{axis}
\end{tikzpicture}%
\caption{Guaranteed upper bounds on the operator gain of the system in Ex.~\ref{ex:2} from noisy input-output trajectories for increasing noise levels $\bar{v}$ and $N=50$ data points. }
\label{fig:ex2_noise}
\end{figure}
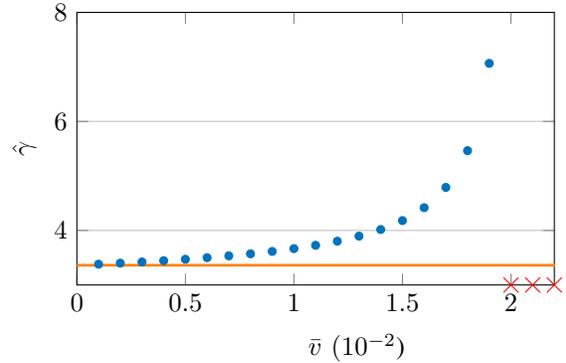
\end{example}

\begin{remark}
Most existing data-based dissipativity analysis approaches only consider dissipativity over the finite time horizon, and more importantly, they cannot provide quantitative guarantees in the case of noisy data. To be more specific, the existing one-shot methods exploiting the Fundamental Lemma (e.g.~\cite{Maupong2017,Romer2019a,Koch2020}) do not provide guarantees under the presence of noise, and the existing iterative methods (e.g.~\cite{Wahlberg2010,Rojas2012,Tanemura2019,Romer2019c,Mueller2017}) at most provide asymptotic guarantees in the case of noisy data. On the other hand, the results presented in this section lead to simple LMI-based conditions which can guarantee dissipativity of an unknown system based on noisy input-output data of finite length. It is part of future research to extend the introduced approach to find a tight description of the system properties given conservative upper bounds on the lag.
\end{remark}

\section{Experimental application example}
In the following, we apply the presented results to an experimental setup to show the potential and applicability of the introduced ideas in real-world applications. More specifically, 
we apply the result from Sec.~\ref{sec:noise} to determine bounds on the operator gain as well as passivity properties of a two-tank system locally around a steady state, and we compare the results to system identification approaches.

The experimental setup of the two-tank water system can be seen in Fig.~\ref{fig:twotank}. It consists of two identical water tanks. The first water tank is fed by a water pump and the second water tank is fed by an outlet of the first water tank and has a water outlet itself. The considered input is the voltage $u_v$ which directly influences the throughput of the pump $v$. The heights of the two tanks $h_1$ and $h_2$ are considered our outputs, which can be measured. 

By first principles, the two-tank can be modeled by 
\begin{align*}
\dot{h_1} &= - \frac{O_1}{A_1} \mu_1 \sqrt{2gh_1} + \frac{1}{A_1} k_\text{p} u_v \\
\dot{h_2} &= - \frac{O_2}{A_2} \mu_2 \sqrt{2gh_2}  + \frac{O_1}{A_2} \mu_1 \sqrt{2gh_1} 
\end{align*}
where $A_i$, $O_i$ are the cross section and the outlet of tank $i=1,2$, respectively, $k_\text{p}$ is a constant of the pump including the tube and outlet and $\mu_i$ captures approximately the hydro-dynamic resistance of the outlet of tank $i$, $i=1,2$.
 
Linearizing the system around a stationary point $(u_v^0, h_1^0, h_2^0)$ yields
\begin{align}
\begin{split}
\dot{x} &= \begin{pmatrix} 
-\frac{\mu_1 O_1 \sqrt{2g}}{2 A_1 \sqrt{h_1^0}} & 0 \\
\frac{\mu_1 O_1 \sqrt{2g}}{2 A_2 \sqrt{h_1^0}} & - \frac{\mu_2 O_2 \sqrt{2g}}{2 A_2 \sqrt{h_2^0}}
 \end{pmatrix} x + \begin{pmatrix} \frac{k_\text{p}}{A_1} \\ 0 \end{pmatrix} u, \\
y &= \begin{pmatrix} 1 &0 \\ 0 &1 \end{pmatrix} x,
\end{split}
\label{eq:lin}
\end{align}
with $u = u_v - u_v^0$, $x_1 = h_1 - h_1^0$ and $x_2 = h_2 - h_2^0$.
\begin{figure}[t]
\centering
\includegraphics[width=0.22\textwidth]{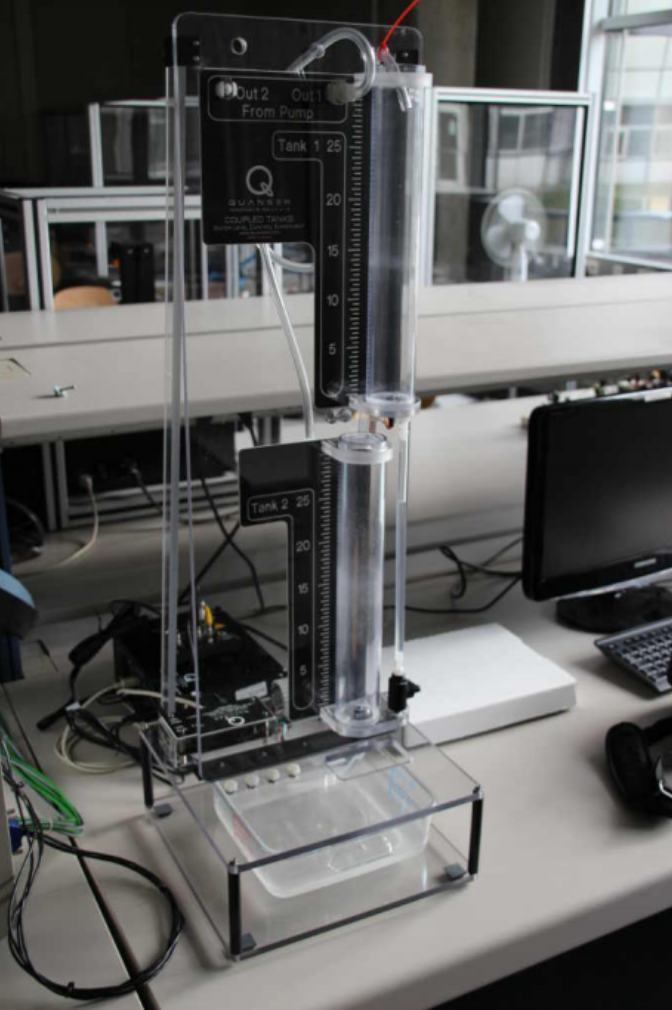}
\includegraphics[width=0.23\textwidth]{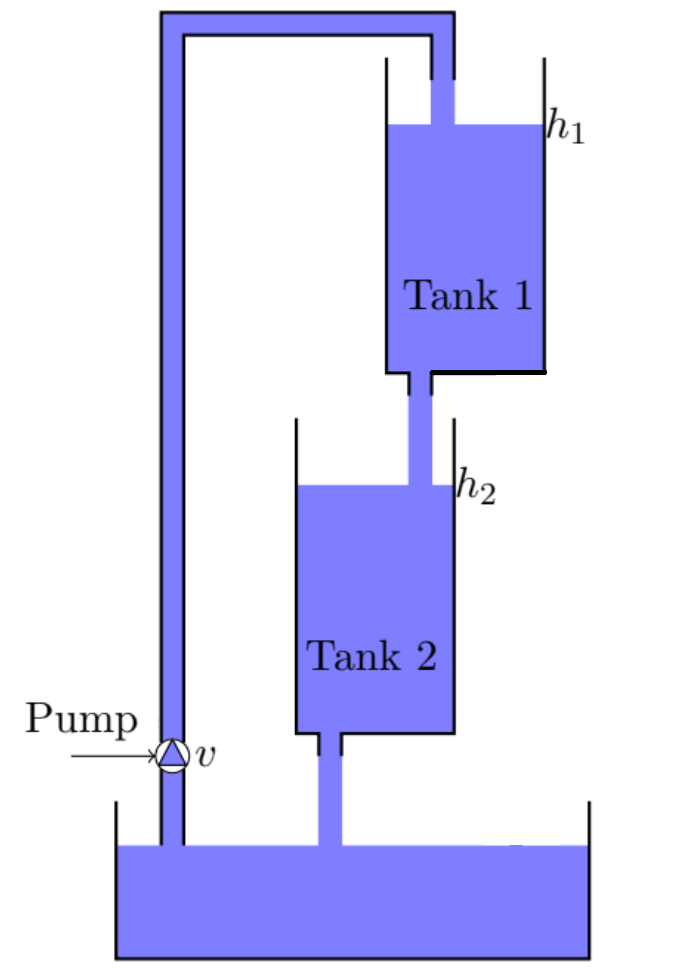}
\caption{Experimental setup and schematic of the two-tank system.}
\label{fig:twotank}
\end{figure}
We are now interested in local information about the system around a setpoint. The determine a setpoint, we apply a simple controller in a first experiment that stabilizes the systems at $h_2^0$. From this first experiment, we approximately determine the steady state $(u_v^0, h_1^0, h_2^0) = (6.8, 13.8, 16.4)$. 
In a second experiment, we excite our system around this steady state. 
The resulting input signal $u_v$ can be seen in Fig.~\ref{fig:data_wt}.
We measure the resulting two heights $h_1$ and $h_2$ over 18 seconds with a sampling time of $T_\text{s}= 0.4$ seconds. 
The measured heights $h_1$ and $h_2$ are also shown in Fig.~\ref{fig:data_wt}.
\begin{figure}
\definecolor{mycolor1}{rgb}{0.00000,0.44700,0.74100}%
\definecolor{mycolor2}{rgb}{0.85000,0.32500,0.09800}%
\definecolor{mycolor3}{rgb}{0.92900,0.69400,0.12500}%
\begin{tikzpicture}
\pgfplotsset{set layers}
\begin{axis}[%
width=0.3\textwidth,
height=0.2\textwidth,
at={(0cm,0cm)},
scale only axis,
axis y line*=left,
xmin=0,
xmax=18,
xlabel={time [s]},
ymin=0,
ymax=25,
ylabel={water height [cm]},
ylabel style = {align=center},
ylabel={Voltage $u_v$ [V] \ref{pgfplots:plot1}},
axis background/.style={fill=white},
ymajorgrids,
]
\addplot [color=mycolor1]
  table[row sep=crcr]{%
0	13.2655764670235\\
0.400000000000006	12.6632892916176\\
0.800000000000011	12.1253633316554\\
1.20000000000002	11.65267931736\\
1.59999999999999	11.2279890692893\\
2	10.820832363288\\
2.40000000000001	10.4603786745044\\
2.80000000000001	10.1199511078344\\
3.20000000000002	9.81534892705464\\
3.59999999999999	9.51840659522308\\
4	9.28442188496268\\
4.40000000000001	9.07564867139603\\
4.80000000000001	8.86206497777304\\
5.20000000000002	8.67056696267767\\
5.59999999999999	8.51860947123724\\
6	8.35728568242636\\
6.40000000000001	8.22688144987772\\
6.80000000000001	8.12935129544252\\
7.20000000000002	8.01773602251826\\
7.59999999999999	7.92346982644101\\
8	7.82683771160467\\
8.40000000000001	7.76276552322453\\
8.80000000000001	7.71668179170251\\
9.20000000000002	7.65019803470887\\
9.59999999999999	7.60549317238448\\
10	7.55092205891082\\
10.4	7.52404391132975\\
10.8	7.49531338610591\\
11.2	7.45753127189237\\
11.6	7.44341031386354\\
12	7.41832504744813\\
12.4	7.39429820379769\\
12.8	7.38651248992973\\
13.2	7.37593309249175\\
13.6	7.36449671540713\\
14	7.36004186939674\\
14.4	7.36015879556359\\
14.8	7.34244264711988\\
15.2	7.34286961582916\\
15.6	7.33241773072604\\
16	7.34159004766602\\
16.4	7.33194529932428\\
16.8	7.33033120752742\\
17.2	7.34114177617879\\
17.6	7.33357681742052\\
18	7.32865703080828\\
};
\label{pgfplots:plot1}
\end{axis}

\begin{axis}[
width=0.3\textwidth,
height=0.2\textwidth,
at={(0cm,0cm)},
scale only axis,
scale only axis,
xmin=0,
xmax=18,
ymin=0,
ymax=25,
axis y line*=right,
axis x line=none,
ylabel style = {align=center},
ylabel={tank height $h_1$ [cm] \ref{pgfplots:plot2} \\ tank height $h_2$ [cm] \ref{pgfplots:plot3}},
]
\addplot [color=mycolor2, thick, dashed]
  table[row sep=crcr]{%
0	15.0687656027159\\
0.400000000000006	15.6042326702713\\
0.800000000000011	16.0741580742456\\
1.20000000000002	16.4720355622911\\
1.59999999999999	16.8114491831535\\
2	17.1347180413618\\
2.40000000000001	17.4101270177453\\
2.80000000000001	17.6525290693512\\
3.20000000000002	17.8568101307613\\
3.59999999999999	18.0454046719257\\
4	18.1818261482487\\
4.40000000000001	18.2964664736786\\
4.80000000000001	18.4106600029467\\
5.20000000000002	18.5003993351067\\
5.59999999999999	18.5676842676238\\
6	18.6265287997206\\
6.40000000000001	18.6657163564346\\
6.80000000000001	18.6719350235492\\
7.20000000000002	18.6806710179875\\
7.59999999999999	18.7014902389883\\
8	18.6990614345562\\
8.40000000000001	18.6755281810567\\
8.80000000000001	18.6414341570509\\
9.20000000000002	18.6111222228511\\
9.59999999999999	18.581916083153\\
10	18.5624240963984\\
10.4	18.5151156831204\\
10.8	18.4678343328671\\
11.2	18.4317177458192\\
11.6	18.3816407992314\\
12	18.341866782604\\
12.4	18.2961535549763\\
12.8	18.2514788565656\\
13.2	18.2020142357686\\
13.6	18.1544092709048\\
14	18.108525256728\\
14.4	18.060780002268\\
14.8	18.0171954821637\\
15.2	17.9696760171956\\
15.6	17.9233967165622\\
16	17.8770856101762\\
16.4	17.8336771515672\\
16.8	17.7909503419435\\
17.2	17.7474306563317\\
17.6	17.7062214583277\\
18	17.6663447660066\\
};
\label{pgfplots:plot2}

\addplot [color=mycolor3, thick, dash dot]
  table[row sep=crcr]{%
0	15.9003385313932\\
0.400000000000006	15.9392671917399\\
0.800000000000011	15.9816401808596\\
1.20000000000002	16.0326308957868\\
1.59999999999999	16.0949539521915\\
2	16.1566522688505\\
2.40000000000001	16.2211243485861\\
2.80000000000001	16.2981953094937\\
3.20000000000002	16.3786809008706\\
3.59999999999999	16.4667834017441\\
4	16.5473441193566\\
4.40000000000001	16.6256965570803\\
4.80000000000001	16.7086766849606\\
5.20000000000002	16.7946330856594\\
5.59999999999999	16.8664281784025\\
6	16.9541521239027\\
6.40000000000001	17.0327198975406\\
6.80000000000001	17.1125825781376\\
7.20000000000002	17.2024990913953\\
7.59999999999999	17.2661557973211\\
8	17.3531335001608\\
8.40000000000001	17.4308390382821\\
8.80000000000001	17.5024177479873\\
9.20000000000002	17.5884333938838\\
9.59999999999999	17.6543357507591\\
10	17.7200144800915\\
10.4	17.786823870399\\
10.8	17.8552332882046\\
11.2	17.9213788439176\\
11.6	17.9795287271892\\
12	18.037880043365\\
12.4	18.1007306935185\\
12.8	18.148386817817\\
13.2	18.2028707336554\\
13.6	18.2564068306478\\
14	18.3022465440513\\
14.4	18.3457772008555\\
14.8	18.4011472328359\\
15.2	18.4441996012778\\
15.6	18.4956606874894\\
16	18.5299380362255\\
16.4	18.5780675736528\\
16.8	18.6185303095776\\
17.2	18.6488197021031\\
17.6	18.6931151617168\\
18	18.733872740508\\
};
\label{pgfplots:plot3}
\end{axis}
\end{tikzpicture}%
\caption{Input $u_v$ and the measured heights $h_1$ and $h_2$ of the two tank system over a horizon of 18 seconds and a sampling time of $T_\text{s} = 0.4$ seconds.}
\label{fig:data_wt}
\end{figure}
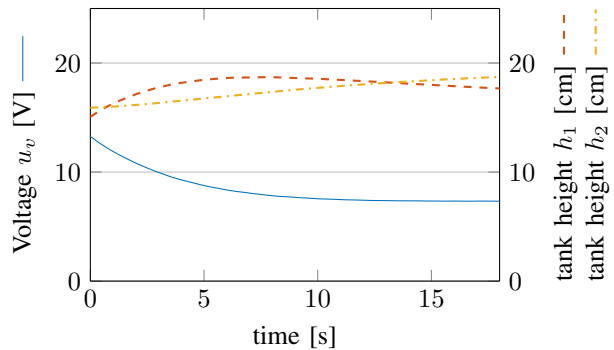

As $g$ is a given physical constant and the parameters $A_i$ and $O_i$ are geometrical specifications of the experimental setup, we estimate the remaining parameters $k_\text{p}$, $\mu_1$ and $\mu_2$ in~\eqref{eq:lin} via a least-squares approach (\textsc{Matlab} function \textit{lsqcurvefit}). The operator norm of the resulting system is given by $\gamma_{\text{ParamEst}} = 6.85$. 

As an alternative, we compute $\begin{pmatrix} A & B \end{pmatrix}$ via
\begin{align}
\begin{pmatrix} A & B \end{pmatrix} = X_+ \begin{pmatrix} X \\ U \end{pmatrix}^\dagger
\label{eq:ls}
\end{align}
where $^\dagger$ denotes the right inverse. The result can be interpreted as the result to the least-squares problem
\begin{align*}
\min_{\begin{pmatrix} A & B \end{pmatrix}} \| X_+ - \begin{pmatrix} A & B \end{pmatrix} \begin{pmatrix} X \\ U \end{pmatrix} \|_F,
\end{align*}
where $F$ indicates the Frobenius norm. 
The resulting operator gain of the identified system is given by $\gamma_{\text{LS}} = 5.20$, underestimating the operator gain $\gamma_{\text{ParamEst}}$ by more than $20 \%$.

Secondly, we are interested in the passivity properties of the system. More specifically, we want to determine the input feedforward passivity index $\rho$ (cf.~\eqref{eq:pi}) for the input-output pair $u_v$ and $h_2$.
The estimated input feedforward passivity index by consecutive system identification and analysis are collected in Tab.~\ref{tab:SysID}, together with the summary of operator gain estimates.
\begin{table}[ht]
 \caption{Consecutive model identification \& analysis}
 \label{tab:SysID}
 \begin{center}
{\renewcommand{\arraystretch}{1.6}%
 \begin{tabular}{|cc|}
 \hline
\pbox{4cm}{\quad ParamEst \quad }  & \pbox{4cm}{\quad LS (Eq.~\eqref{eq:ls}) \quad } \\[1ex]
 \hline
\hline
\pbox{2.5cm}{$\hat{\gamma} = 6.85$} 
& \pbox{2.5cm}{$\hat{\gamma} = 5.20$}  \\[1ex]
 \hline
\pbox{2.5cm}{$\hat{\rho} = -0.515$} 
& \pbox{2.5cm}{$\hat{\rho} = -0.588$} \\[1ex]
 \hline
 \end{tabular}}
 \end{center}
 \end{table} 

Next, we apply the result in Sec.~\ref{sec:noise} to find a guaranteed operator gain as well as a guaranteed input feedforward passivity index for all systems consistent with the data
from noisy state measurements. For this, we assume that process noise enters the linearized model \eqref{eq:lin} as described in \eqref{eq:sys_noise}. We apply the results of Thm.~\ref{thm:noise} by assuming a bound on the process noise given by $\| w_k \|_2 \leq \bar{w} $, which implies $\hat{W} \hat{W}^\top \preceq \bar{w}^2 N I$. The result is plotted in Fig.~\ref{fig:wt} and Fig.~\ref{fig:wt2} for different assumed noise levels $\bar{w}$. 
In fact, we retrieve provably robust, reasonable upper bounds on the operator gain of the system, as well as lower bounds on the input feedforward passivity index. 

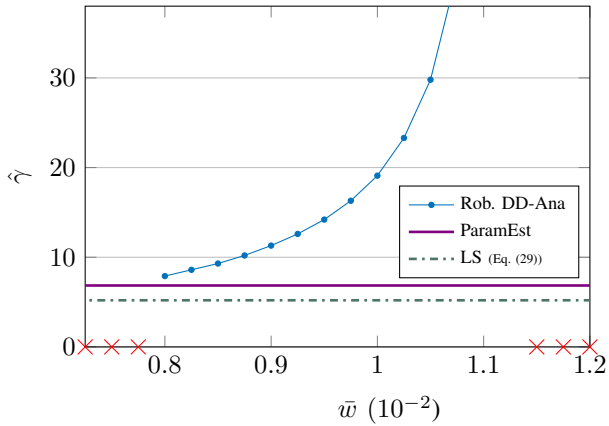
\begin{figure}[h]
\definecolor{mycolor1}{rgb}{0.00000,0.44700,0.74100}%
\definecolor{dgreen}{RGB}{78,117,102}%
\begin{tikzpicture}
\begin{axis}[%
width=0.37\textwidth,
height=0.25\textwidth,
at={(0cm,0cm)},
scale only axis,
xmin=0.00725,
xmax=0.012,
xtick scale label code/.code={$\bar{w}$ $(10^{#1}$) \quad \quad \quad \quad \quad \quad \quad \quad \quad \quad \quad \quad},
ymin=0,
ymax=38,
ylabel={$\hat{\gamma}$},
axis background/.style={fill=white},
ymajorgrids,
legend style={at={(axis cs:0.0119,18)}},
legend cell align={left},
]
\addplot [color=mycolor1,mark size=1pt,mark=*,mark options={solid}]
  table[row sep=crcr]{%
0.008     7.9\\
0.00825   8.6\\
0.0085    9.3\\
0.00875   10.2\\
0.009     11.3\\
0.00925   12.6\\
0.0095    14.2\\
0.00975   16.3\\
0.01      19.1\\
0.01025   23.3\\
0.0105    29.8\\
0.01075   41.7\\
0.011     70.0\\
0.01125   220\\
};
\addlegendentry{\scriptsize Rob.~DD-Ana};
\addplot [color=violet,solid,line width=1pt]
  table[row sep=crcr]{
0	6.85\\
0.012 6.85\\
};
\addlegendentry{\scriptsize ParamEst};
\addplot [color=dgreen,solid,line width=1pt, dash dot]
  table[row sep=crcr]{
0	5.20\\
0.012 5.20\\
};
\addlegendentry{\scriptsize LS \tiny{(Eq.~\eqref{eq:ls})}};

\addplot [only marks,color=red,mark size=4pt,mark=x,mark options={solid},forget plot]
  table[row sep=crcr]{%
	0.00725 0\\
	0.0075 0\\
	0.00775 0\\
	0.0115 0\\
	0.01175 0\\
	0.012 0\\
};
\end{axis}
\end{tikzpicture}%
\caption{Upper bound on the operator gain of the two-tank system from noisy input-state trajectories for different bounds on the measurement noise $\bar{w}$. 
The blue dots \textcolor{mycolor1}{$\bullet$} are the computed upper bounds by the introduced robust data-driven analysis, the red crosses \textcolor{purple}{$\times$} indicate that no upper bound could be found and the differently colored lines indicate the results of consecutive system identification and systems analysis as indicated by the legend.
}
\label{fig:wt}
\end{figure}

\begin{figure}[h]
\definecolor{mycolor1}{rgb}{0.00000,0.44700,0.74100}
\definecolor{rosa}{RGB}{234,137,154}
\definecolor{dgreen}{RGB}{78,117,102}%
\begin{tikzpicture}
\begin{axis}[%
width=0.37\textwidth,
height=0.25\textwidth,
at={(0cm,0cm)},
scale only axis,
xmin=0.00725,
xmax=0.012,
xtick scale label code/.code={$\bar{w}$ $(10^{#1}$) \quad \quad \quad \quad \quad \quad \quad \quad \quad \quad \quad \quad \quad},
ymin=-2.9,
ymax=0.2,
ylabel={$\hat{\rho}$},
axis background/.style={fill=white},
ymajorgrids,
legend style={at={(axis cs:0.0119,-0.7)}},
legend cell align={left},
]
\addplot [color=violet,solid,line width=1pt]
  table[row sep=crcr]{
0	-0.51\\
0.012 -0.51\\
};
\addlegendentry{\scriptsize ParamEst};
\addplot [color=dgreen,solid,line width=1pt, dash dot]
  table[row sep=crcr]{
0	-0.59\\
0.012 -0.59\\
};
\addlegendentry{\scriptsize LS \tiny{(Eq.~\eqref{eq:ls})}};
\addplot [color=mycolor1,mark size=1pt,mark=*,mark options={solid}]
  table[row sep=crcr]{%
0.008   -0.9903\\
0.00825   -1.0528\\
0.0085   -1.1191\\
0.00875   -1.1946\\
0.009   -1.2800\\
0.00925   -1.3778\\
0.0095   -1.4927\\
0.00975   -1.6299\\
0.01   -1.7982\\
0.01025   -2.0116\\
0.0105   -2.2950\\
0.01075   -2.6996\\
0.011   -3.3566\\
0.01125   -4.8399\\
};
\addlegendentry{\scriptsize Rob.~DD-Ana};
\addplot [only marks,color=red,mark size=4pt,mark=x,mark options={solid},forget plot]
  table[row sep=crcr]{%
	0.00725 -2.9\\
	0.0075 -2.9\\
	0.00775 -2.9\\
	0.0115 -2.9\\
	0.01175 -2.9\\
	0.012 -2.9\\
};
\end{axis}
\end{tikzpicture}%
\caption{Lower bound on the input-feedforward passivity index of the two-tank system from noisy input-state trajectories for different bounds on the measurement noise $\bar{w}$.  
The blue dots \textcolor{mycolor1}{$\bullet$} are the computed lower bounds by the introduced robust data-driven analysis, the red crosses \textcolor{purple}{$\times$} indicate that no lower bound could be found and the differently colored lines indicate the results of consecutive system identification and systems analysis as indicated by the legend.
}
\label{fig:wt2}
\end{figure}
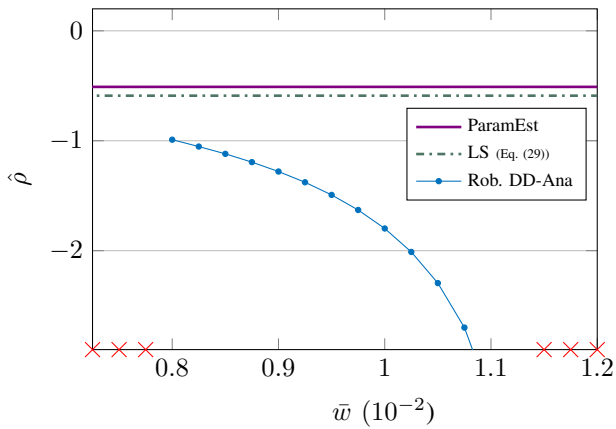

We can see that for an assumed noise bound of $\bar{w} \leq 0.0077$, neither an operator gain nor an input feedforward passivity index can be found. This implies that there exists no LTI system which admits any $\mathcal{L}_2$-gain $\gamma > 0$ or feedforward passivity index $\rho \in\mathbb{R}$, respectively, that is consistent with the data assuming such a low noise level. For an assumed noise level of $\bar{w} = 0.008$ the proposed approach yields an operator gain of $\gamma = 7.92$ and an input feedforward passivity index of $\rho = -0.99$. 
When comparing these values to our baseline \textit{ParamEst}, which uses physical insights together with parameter estimation, the results are more conservative. This is to be expected as the resulting upper (or lower) bound is the system property that is guaranteed for all systems consistent with the data and the assumed noise model. In particular, we note that the approaches based on system identification methods do not provide any theoretical guarantees on the actual system property satisfied by the two-tank system. The chosen noise level of course highly influences the result, as indicated also in Fig.~\ref{fig:wt} and Fig.~\ref{fig:wt2}. However, as too small noise levels lead to infeasibility, the data already implicitly reveals a reasonable interval for the noise bound. 
Naturally, for larger assumed noise bounds, the estimated bound on the system property becomes more conservative, since the set of systems that are consistent with the data and the noise bound increases. Finally, we can also see that for noise bounds larger than $\bar{w} \geq 0.0115$ no bounds on the respective system properties can be found anymore. 

Altogether, the presented results show the potential and applicability of the introduced approach for real experimental measurements. For a given noise level, the results provide provably robust bounds on dissipativity properties over the infinite horizon for all systems consistent with the data. By varying the assumed noise bound, our framework allows for an intuitive trade-off between accuracy of the estimated system property and robustness (i.e., the size of the set of systems for which the property is guaranteed).

\section{Conclusion and Outlook}
In this work, we introduced simple verification methods of dissipativity properties with guarantees from (noisy) input-state and input-output data based on LMIs. 
While in \cite{Koch2020a} the general ideas were presented to determine dissipativity properties that hold over the infinite horizon from finite input and state trajectories, the guarantees given in this paper for noisy input-state data are non-conservative and computationally less expensive.
Furthermore, we introduced approaches to verify dissipativity from input-output data both with and without noise.

Future work includes the problem of improving the conditions for dissipativity properties from noisy input-output data. 
Furthermore, it might be interesting for future work to investigate how existing works on data-driven descriptions of nonlinear systems (Hammerstein and Wiener systems, second-order Volterra systems, bilinear systems, and polynomial systems) as presented in \cite{Berberich2019c,Rueda2020,Bisoffi2020,Guo2020}, respectively, can be applied to verify and find dissipativity properties of unknown nonlinear systems from data.
\appendix[Proof of Lem.~\ref{lem:extended}]
\begin{proof}
The input-output behavior of the system $G$ in \eqref{eq:sys} over $l$ steps can be written as \eqref{eq:input-output},
\begin{figure*}
\vspace{2pt}
\begin{align}
\begin{pmatrix} y_{k-l} \\ y_{k-l+1} \\ \vdots \\ y_{k-1} \end{pmatrix} = \underbrace{\begin{pmatrix} C \\ CA \\ \vdots \\ CA^{l-1} \end{pmatrix}}_{\mathcal{O}_l} x_{k-l} + 
\underbrace{\begin{pmatrix} 
D & 0 & & \dots & 0 & 0 \\
CB & D & & \dots & 0 & 0 \\
\vdots & \ddots & \ddots&\ddots&\ddots&\vdots \\
CA^{l-2}B & CA^{l-3}B &\dots & CAB & CB &D
\end{pmatrix} }_R
\begin{pmatrix} u_{k-l} \\ u_{k-l+1} \\ \vdots \\ u_{k-1} \end{pmatrix}
\label{eq:input-output}
\end{align}
\begin{align}
\begin{split}
\label{eq:extended_sys}
\underbrace{\begin{pmatrix} u_{k-l+1} \\ \vdots \\ u_{k-1} \\ u_{k} \\ y_{k-l+1} \\ \vdots \\ y_{k-1} \\ y_{k}  \end{pmatrix}}_{\xi_{k+1}} = \underbrace{\left(\begin{pmatrix} 
0 & I & \dots & 0 & 0 &  0& \dots &0 \\
\vdots & \ddots & \ddots&\ddots& \vdots & & \ddots&\vdots \\
0 & 0 & \dots & I & 0& 0&  \dots & 0 \\
0 & 0 & \dots & 0 & 0 & 0&  \dots & 0 \\
0 & 0 & \dots & 0 & 0& I& \dots &0 \\
\vdots & \ddots & \ddots&\vdots&\vdots & & \ddots&\vdots \\
0 & 0 & \dots & 0 & 0& 0&  \dots & I \\
CA^{l-1}B & \dots & \dots & CB & 0 & 0 & \dots &  0 \\
\end{pmatrix} + \begin{pmatrix} 0 \\ \vdots \\ 0\\ 0 \\ 0 \\ \vdots \\ 0 \\ C A^l T \end{pmatrix}\right) }_{\widetilde{A}} \underbrace{\begin{pmatrix} u_{k-l} \\ u_{k-l+1} \\ \vdots \\ u_{k-1} \\ y_{k-l} \\ y_{k-l+1} \\ \vdots \\ y_{k-1} \end{pmatrix}}_{\xi_{k}} + \underbrace{\begin{pmatrix} 0 \\ \vdots \\ 0 \\ I \\ 0 \\ \vdots \\ 0 \\ D  \end{pmatrix}}_{\widetilde{B}} u_{k} \\
y_k = \underbrace{\begin{pmatrix} 0 & \dots & 0 & I \end{pmatrix} \tilde{A}}_{\tilde{C}} \xi_k + \underbrace{D}_{\tilde{D}} u_k
\end{split} 
\end{align}
\noindent\makebox[\linewidth]{\rule{\textwidth}{0.4pt}}
\end{figure*}
which yields with the introduced matrix notation
\begin{align*}
\begin{pmatrix} 
-R & I
\end{pmatrix} \xi_k = \mathcal{O}_l x_k.
\end{align*}
Using the definition of the lag, we know that $\mathcal{O}_l$ has full column rank, and hence, there exists a left-inverse $\mathcal{O}_l^{-1}$ (which has full row rank) such that
\begin{align}
\underbrace{\mathcal{O}_l^{-1} \begin{pmatrix} 
-R & I
\end{pmatrix}}_{T} \xi_k = x_k.
\label{eq:trafo}
\end{align}

The general system description from \eqref{eq:sys} yields
\begin{align*}
y_{k} = C A^l x_{k-l} + \begin{pmatrix} CA^{l-1} B & \dots & CB \end{pmatrix} \begin{pmatrix} u_{k-l} \\ \vdots \\ u_{k-1} \end{pmatrix}.
\end{align*}
Together with $T$ as defined in \eqref{eq:trafo}, 
this leads to \eqref{eq:extended_sys}. This proves that $(\tilde{A}, \tilde{B}, \tilde{C}, \tilde{D})$ can explain the input-output trajectory. 
\end{proof} 
\section*{Acknowledgment}
The authors thank Said Jamal Mohamad for his help in performing the experiments with the two-tank system.
\ifCLASSOPTIONcaptionsoff
  \newpage
\fi
\bibliographystyle{IEEEtran}
\bibliography{../bib_all}

\end{document}